\documentclass[a4paper,USenglish,thm-restate,modulo]{lipics-v2021}
\hideLIPIcs

\usepackage{xspace}

\newcommand{\from}{\colon\xspace}

\newcommand{\reals}{\mathbb{R}\xspace}
\newcommand{\naturals}{\mathbb{N}\xspace}
\newcommand{\T}{\mathcal{T}\xspace}

\newcommand{\Tr}{T_\mathrm{root}\xspace}

\renewcommand{\L}{\Lambda^{\phantom{+}}\xspace}
\newcommand{\E}{\Lambda^+\xspace}
\newcommand{\bigO}{\mathcal{O}}
\newcommand{\etal}{\textit{et al.}\xspace}
\newcommand{\emb}{\ensuremath{\mathcal D}\xspace}
\newcommand{\W}{\ensuremath{\mathcal W}\xspace}
\newcommand{\Wt}{\ensuremath{\widetilde{\W}}\xspace}

\renewcommand{\paragraph}[1]{\smallskip\noindent\textbf{#1.}}
\newcommand{\mypar}[1]{{\smallskip\noindent\sffamily\normalsize\bfseries{#1}\hspace{1ex}}}

\usepackage{nameref}
\makeatletter
\let\orgdescriptionlabel\descriptionlabel
\renewcommand*{\descriptionlabel}[1]{%
  \let\orglabel\label
  \let\label\@gobble
  \phantomsection
  \edef\@currentlabel{#1}%
  \let\label\orglabel
  \orgdescriptionlabel{#1}%
}
\makeatother

\graphicspath{{./figures/}}

\title{Polygon-Universal Graphs}

\author{Tim Ophelders}{Department of Mathematics and Computer Science, TU Eindhoven, The Netherlands}{t.a.e.ophelders@tue.nl}{}{}
\author{Ignaz Rutter}{Department of Computer Science and Mathematics, University of Passau, Germany}{rutter@fim.uni-passau.de}{https://orcid.org/0000-0002-3794-4406}{Partially supported by the German Science Foundation (DFG); Ru 1903/3-1.}
\author{Bettina Speckmann}{Department of Mathematics and Computer Science, TU Eindhoven, The Netherlands}{b.speckmann@tue.nl}{https://orcid.org/0000-0002-8514-7858}{Partially supported by the Dutch Research Council (NWO); 639.023.208.}
\author{Kevin Verbeek}{Department of Mathematics and Computer Science, TU Eindhoven, The Netherlands}{k.a.b.verbeek@tue.nl}{}{}

\authorrunning{T. Ophelders and I. Rutter and B. Speckmann and K. Verbeek} 

\Copyright{Tim Ophelders and Ignaz Rutter and Bettina Speckmann and Kevin Verbeek} 

\ccsdesc[500]{Human-centered computing~Graph drawings}

\keywords{Graph drawing, partial drawing extension, simple polygon} 

\relatedversion{} 

\supplement{}

\funding{Research on the topic of this paper was initiated at the 2nd Workshop on Applied Geometric Algorithms in Vierhouten, NL, supported by the Dutch Research Council (NWO); 639.023.208. }

\acknowledgements{Ignaz Rutter would like to thank Michael Hoffmann and Vincent Kusters for discussions on conjectures related to this paper.}

\nolinenumbers 

\begin{document}
\maketitle

\begin{abstract}
    We study a fundamental question from graph drawing: given a pair $(G,C)$ of a graph $G$ and a cycle~$C$ in~$G$ together with a simple polygon~$P$, is there a straight-line drawing of $G$ inside $P$ which maps $C$ to $P$? We say that such a drawing of $(G,C)$ \emph{respects}~$P$. We fully characterize those instances $(G,C)$ which are \emph{polygon-universal}, that is, they have a drawing that respects $P$ for any simple (not necessarily convex) polygon $P$. Specifically, we identify two necessary conditions for an instance to be polygon-universal. Both conditions are based purely on graph and cycle distances and are easy to check. We show that these two conditions are also sufficient. Furthermore, if an instance $(G,C)$ is planar, that is, if there exists a planar drawing of $G$ with $C$ on the outer face, we show that the same conditions guarantee for every simple polygon $P$ the existence of a planar drawing of $(G,C)$ that respects $P$. If $(G,C)$ is polygon-universal, then our proofs directly imply a linear-time algorithm to construct a drawing that respects a given polygon $P$.
\end{abstract}

\section{Introduction}

Graphs are a convenient way to express relations between entities. To visualize these relations, the corresponding graph needs to be \emph{drawn}, most commonly in the plane and with straight edges. Naturally there are a multitude of different optimization criteria and drawing restrictions that attempt to capture various perceptual requirements or real-world conditions. In this paper we focus on drawings which are constrained to the interiors of simple polygons.

The \emph{polygon-extension problem} asks, whether a given graph admits a (planar) drawing where the outer face is fixed to a given simple polygon~$P$; see Fig.~\ref{fig:respectingExamples} for examples.  Our main focus is the \emph{polygon-universality problem}, which asks whether a given plane graph admits a polygon-extension for every choice of fixing the outer face to a simple polygon.
As is often the case with geometric problems, a natural complexity class for the polygon-extension problem is~$\exists \mathbb R$, the class of problems that can be encoded in polynomial time as an existentially quantified formula of real variables (rather than Boolean variables as for {\sc Sat}), which was introduced by  Schaefer and \v{S}tefankovi\v{c}~\cite{ss-fpnee-15}.  The natural complexity class for the polygon-universality problem is~$\forall\exists \mathbb R$, the universal existential theory of the reals, which has been recently defined by Dobbins~\etal~\cite{dkmr-aerca-18}. It is known that $\mathrm{NP} \subseteq \exists \mathbb R \subseteq \forall\exists\mathbb R \subseteq \mathrm{PSPACE}$~\cite{c-agcp-88}.

Tutte~\cite{t-hdg-63} proved that there is a straight-line planar drawing of a planar graph $G$ inside an arbitrary convex polygon $P$ if one fixes the outer face of (an arbitrary planar embedding of) $G$ to $P$.
This result has been generalized to allow polygons $P$ that are non-strictly convex~\cite{cegl-dgppo-12,dgk-pdhgg-11} or even star-shaped polygons~\cite{hm-cdgnb-08}. 
These results have applications in \emph{partial drawing extension} problems. Here, in addition to an input graph $G$, we are given a subgraph~$H \subseteq G$ together with a fixed drawing~$\Gamma$ of $H$. The question is whether one can extend the given drawing $\Gamma$ to a planar straight-line drawing of the whole graph $G$ by drawing the vertices and edges of $G-H$ inside the faces of $H$. If the embedding of $G$ is fixed, the results by Tutte and others allow to reduce the problem by removing vertices of $G$ that are contained in convex or star-shaped faces of $\Gamma$. Such reduction rules have lead to efficient testing algorithms for special cases, for example, when the drawing of $H$ is convex~\cite{mnr-ecpdg-16} (note that this is non-trivial, as the outer face is not convex).

Recently, Lubiw~\etal~\cite{lmm-cdgpr-18} showed that it is $\exists \mathbb{R}$-complete to decide for a given planar graph that is partially fixed to a non-crossing polygon \emph{with holes}, whether the partial drawing can be extended to a planar straight-line drawing that does not intersect the outside of the polygon.  That is, the planar polygon extension problem is $\exists \mathbb R$-complete for polygons with holes.  They leave the case of simple polygons open. 

If we do not insist on straight-line drawings, then other questions arise. Angelini~\etal~\cite{aklss-obdog-20} give an $O(mn)$-time algorithm for testing whether an $n$-vertex outer-planar graph admits a planar one-bend drawing whose outer face is fixed to a simple polygon on $m$ vertices.
Mchedlidze and Urhausen~\cite{mu-bsedc-18} link the number of bends per edge that are necessary for extending a drawing to a convexity measure for the faces of the partial drawing (note that it might be impossible to extend a drawing in a planar way even with arbitrarily many bends per edge). Angelini~\etal~\cite{adfjk-tppeg-15} present a linear time algorithm to decide if a partially fixed drawing has a planar drawing using Jordan arcs, while Jel{\'i}nek~\etal~\cite{jkr-ktppe-13} characterize the solvable instances by forbidden substructures.

Quite recently, Dobbins~\etal~\cite{dkmr-aerca-18} considered the problem of area-universality for graphs with partial drawings. Let $G$ be a planar graph with a fixed embedding, including the outer face.  An assignment of areas to faces of $G$ is \emph{realizable} if $G$ admits a straight-line drawing such that each face has the assigned area.  A graph is \emph{area-universal} if every area assignment is realizable.  Dobbins~\etal prove that it is $\exists \mathbb{R}$-complete to decide whether an area assignment is realizable for a planar triangulation, which has been partially drawn, and testing area-universality is $\forall\exists \mathbb{R}$-complete if the planarity condition is dropped (but still parts of the drawing are fixed). They conjecture that the same area-universality problems without a partially fixed drawing are $\exists \mathbb{R}$- and $\forall \exists \mathbb{R}$-complete in general. 

\subparagraph{Notation}
Let~$G=(V,E)$ be a graph with $n$ vertices. A \emph{drawing} $\emb$ of $G$ is a map from each $v \in V$ to points in the plane and from each edge $e \in E$ to a Jordan arc connecting its endpoints. A \emph{straight-line} drawing maps each edge to a straight line segment. A drawing is \emph{planar}, if no two edges intersect, except at common endpoints. A graph $G$ is planar if it has a planar drawing. Let~$C=[c_1,\dots,c_t]$ with~$c_i\in V$ be a simple cycle in~$G$. An \emph{instance}~$(G,C)$ is planar if $G$ has a planar drawing with $C$ as the outer face.
Let $P$ be a simple polygon with~$t$ vertices~$[p_1,\dots,p_t]$ with~$p_i\in\reals^2$. A drawing $\emb$ of~$(G,C)$ \emph{respects} $P$ if it is a map~$\emb\from V\to P$ from vertices to points in~$P$ such that~$\emb(c_i)=p_i$ and for each edge~$\{u,v\}\in E$, the line segment between~$\emb(u)$ and~$\emb(v)$ lies in~$P$ (see Figure~\ref{fig:respectingExamples}). That is,~$\emb$ is a straight-line drawing of~$G$ inside~$P$ that fixes the vertices of~$C$ to the corresponding vertices of~$P$. An instance~$(G,C)$ is \emph{(planar) polygon-universal} if it admits a (planar) straight-line drawing that respects every simple (not necessarily convex) polygon~$P$ on~$t$ vertices. 

\begin{figure}[t]\centering
    \includegraphics{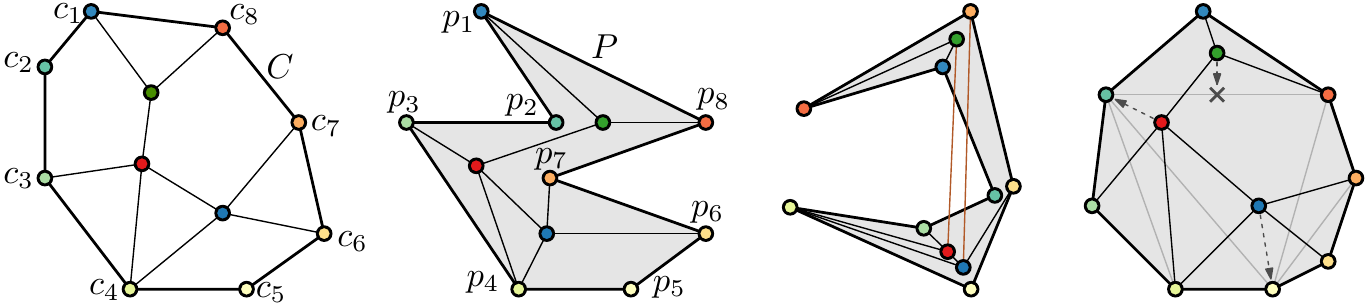}
    \caption{
    Left: an instance $(G,C)$.
    Center left: a drawing of $(G,C)$ that respects a polygon $P$ (shaded in grey).
    Center right: there is no drawing of $(G,C)$ that respects this polygon.
    Right: A triangulated convex polygon with a drawing that is not triangulation-respecting; moving the vertices along the dashed arrows results in a triangulation-respecting drawing.}
    \label{fig:respectingExamples}
\end{figure}

Our algorithms use a triangulation $\T$ of $P$ to construct a drawing or prove the non-universality. We say that a drawing of $(G,C)$ \emph{respects} $\T$ if no edge of $G$ properly crosses an edge of $\T$.
Although every triangulation-respecting drawing is also a drawing, the converse is not true. In fact, there are graphs $G$, cycles $C$, and polygons $P$ that have a drawing, but no triangulation of $P$ exists that allows a triangulation-respecting drawing (see Figure~\ref{fig:nonrespecting}).

\begin{figure}[h]\centering
  \includegraphics{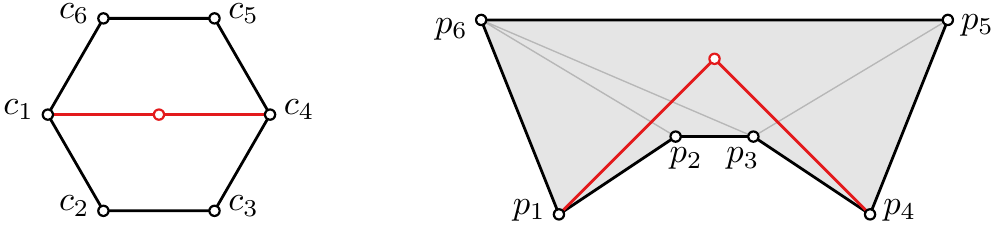}
  \caption{A graph $G$ and a polygon for which there exists a drawing of $G$, but no triangulation with a triangulation-respecting drawing.}
  \label{fig:nonrespecting}
\end{figure}

\subparagraph{Results and organization} In Section~\ref{sec:necessary} we identify two necessary conditions for an instance~$(G,C)$ to be \emph{polygon-universal}. These conditions are purely based on graph and cycle distances and hence easy to check. To show that these two conditions are also sufficient, we use triangulation-respecting drawings: if there is a triangulation $\T$ of a simple polygon $P$ such that $(G,C)$ does not admit a triangulation-respecting drawing for $\T$, then we can argue that $G$ contains one of two forbidden substructures, violating the necessary conditions (see Section~\ref{sec:universality}). These substructures certify that $(G,C)$ is not polygon-universal. 

To arrive at this conclusion, in Section~\ref{sec:triangulation} we first present an algorithm that tests in linear time for a given instance $(G,C)$, a polygon $P$, and triangulation $\T$ of~$P$, whether there exists a triangulation-respecting drawing of $G$ for $\T$ inside $P$. If so, we can construct the drawing in linear time. Then, in Section~\ref{sec:planar}, we consider planar instances $(G,C)$ and show that the same algorithm can decide in linear time whether there is a triangulation-respecting drawing that is planar after infinitesimal perturbation. An analysis of this algorithm shows that a planar instance $(G,C)$ is planar polygon-universal if and only if it is polygon-universal.

Omitted proofs, marked with ($\star$), can be found in the appendix.

\section{Necessary conditions for polygon-universality}\label{sec:necessary}

We present two necessary conditions for an instance $(G,C)$ to be polygon-universal. Intuitively, both conditions capture the fact that there need to be ``enough'' vertices in $G$ between cycle vertices for the drawing not to become ``too tight''. The \emph{\ref{cond:banana} Condition} captures this fact for any two vertices on the cycle $C$. The \emph{\ref{cond:ninja} Condition} is a bit more involved: even if the Pair Condition is satisfied for any pair of vertices on the cycle, there can still be triples of vertices which together ``pull too much'' on the graph. Specifically, for an instance $(G,C)$ of a graph $G$ and a cycle $C\subset G$ with $t$ vertices, we denote by~$d_G\from V\times V\to\naturals$ the graph distance in $G$ and by~$d_C\from V(C)\times V(C)\to\naturals$ the distance (number of edges) along the cycle~$C$.
The following conditions are necessary for $(G,C)$ to be polygon-universal for all simple polygons~$P$: 
  
\begin{description}
  \item[Pair\label{cond:banana}] For all~$i$ and~$j$, we have~$d_C(c_i,c_j)\leq d_G(c_i,c_j)$ (and hence $d_C(c_i,c_j)=d_G(c_i,c_j)$).
  \item[Triple\label{cond:ninja}] For all vertices~$v\in V$ and distinct~$i,j,k$ with~$d_C(c_i,c_j)+d_C(c_j,c_k)+d_C(c_i,c_k)\geq t$ (and hence~$=t$), 
  we have~$d_G(c_i,v)+d_G(c_j,v)+d_G(c_k,v)>t/2$.
\end{description}

\noindent
To establish that these two conditions are necessary, we use the \emph{link distance} between two points inside certain simple polygons $P$. 
Specifically, the link distance of two points $q_1$ and $q_2$ with respect to a simple polygon $P$ is the minimum number of segments for a polyline $\pi$ that lies inside $P$ and connects $q_1$ and $q_2$.
If the \ref{cond:banana} Condition is violated for two cycle vertices $c_i$ and $c_j$, we can construct a \emph{Pair Spiral} polygon $P$ (see Figure~\ref{fig:necessary} (left)) such that the link distance between $p_i$ and $p_j$ (the vertices of $P$ to which $c_i$ and $c_j$ are mapped) exceeds~$d_G(c_i,c_j)$. Clearly there is no drawing $(G,C)$ that respects $P$.

If the first condition holds, but the second condition is violated by a vertex $v$, consider the shortest paths via $v$ that connect $c_i,c_j,c_k$ to each other.  By assumption the total length of these three paths is $2d_G(c_i,v)+2d_G(c_j,v)+2d_G(c_k,v) \le t$, while the total length of the paths connecting $c_i,c_j,c_k$ to each other along $C$ is $d_C(c_i,c_j)+d_C(c_j,c_k)+d_C(c_i,c_k) = t$.  Since the pair condition holds, the paths via $v$ are not shorter than the paths along $C$, and therefore the paths via $v$ must be shortest paths connecting the pairs.
That is,~$d_C(c_i,c_j)=d_G(c_i,v)+d_G(c_j,v)$, $d_C(c_j,c_k)=d_G(c_j,v)+d_G(c_k,v)$ and~$d_C(c_i,c_k)=d_G(c_i,v)+d_G(c_k,v)$.
In that case, we can construct a \emph{Triple Spiral} polygon $P$ (see Figure~\ref{fig:necessary} (right)) such that there is no point that lies within link-distance~$d_G(c_i,v)$ from~$c_i$, link-distance~$d_G(c_j,v)$ from~$c_j$, and link-distance~$d_G(c_k,v)$ from~$c_k$ simultaneously.
Hence, there exists no drawing of the aforementioned shortest paths via $v$ that respects $P$.

\begin{figure}[h]\centering
    \includegraphics{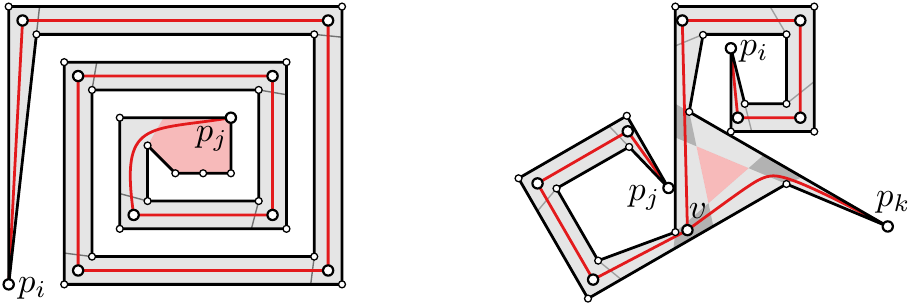}
    \caption{Left: Pair Spiral. Points with link-distance greater than~$d_G(c_i,c_j)$ from~$p_i$ shaded red. Right: Triple Spiral. Points of link-distance $\leq d_G(c_x,v)$ from~$p_x$ for one~$x\in\{i,j,k\}$ in light gray; for two $x\in\{i,j,k\}$ in dark gray; there is no point $q$ in $P$ with $d_G(c_x,q)\leq d_G(c_x,v)$ for all $x\in\{i,j,k\}$.}
    \label{fig:necessary}
\end{figure}

\section{Triangulation-respecting drawings}\label{sec:triangulation}

In this section we are given the following input: an instance $(G,C)$ consisting of a graph $G$ with $n$ vertices and a cycle $C$ with $t$ vertices, and a simple polygon $P$ with $t$ vertices together with an arbitrary triangulation $\T$ of $P$. We study the following question: is there a drawing of $(G,C)$ that respects both $P$ and $\T$? 

We describe a dynamic programming algorithm which can answer this question in linear time. The basic idea is as follows: every edge of $\T$ defines a \emph{pocket} of $P$. We recursively \emph{sketch} a drawing of $G$ within each pocket. Such a sketch assigns an approximate location, such as an edge or a triangle, to each vertex. Ultimately we combine the location constraints on vertex positions posed by the sketches and decide if they can be satisfied.

We root (the dual tree of) $\T$ at an arbitrary triangle $\Tr$. Each edge~$e$ of $\T$ partitions~$P$ into two regions, one of which contains~$\Tr$.
Let~$Q$ be the region not containing~$\Tr$. We say that $Q$ is a \emph{pocket} with the \emph{lid} $e=e_Q$, and we denote the unique triangle outside $Q$ adjacent to $e_Q$ by $T_Q^+$.
Since a pocket is uniquely defined by its lid, we will for an edge $e$ also write $Q_e$ to denote the pocket with lid $e$.
We say that a pocket is \emph{trivial} if its lid lies on the boundary of $P$; in such case the pocket consists of only that edge.
If $Q$ is a non-trivial pocket, then we denote the unique triangle inside $Q$ adjacent to $e_Q$ by $T_Q$ (see Figure~\ref{fig:triangulation}).

For ease of explanation we consider all indices on $C$ and $P$ modulo $t$, that is, we identify $c_i$ and $c_{i+t}$ as well as $p_i$ and $p_{i+t}$.
Moreover, when talking about a non-trivial pocket $Q$ with lid $(p_i,p_j)$, whose third vertex of $T_Q$ is $p_k$, we will assume that $i\leq k\leq j$ (otherwise simply shift the indices cyclically).
We first define triangulation-respecting drawings for pockets: 
\begin{definition}
    A \emph{triangulation-respecting drawing} for a pocket $Q$ with lid $e_Q=(p_i,p_j)$ is an assignment of the vertices of $G$ to locations inside the polygon $P$, such that
    \begin{enumerate}
        \item Any vertex $c_\ell$ with $i \leq \ell \leq j$ is assigned to the polygon vertex $p_\ell$.
        \item For any edge $(u,v)$ of $G$, $u$ and $v$ lie on a common triangle (or edges or vertices thereof).
    \end{enumerate}
    We consider the triangles of the triangulation as closed, so that distinct triangles may share a segment (namely an edge of the triangulation) or a point (namely a vertex of $P$).
    We define a triangulation-respecting drawing for the entire triangulation analogously, requiring that~$c_\ell$ is assigned to~$p_\ell$ for all~$\ell$.
\end{definition}

\begin{figure}[b]\centering
  \includegraphics{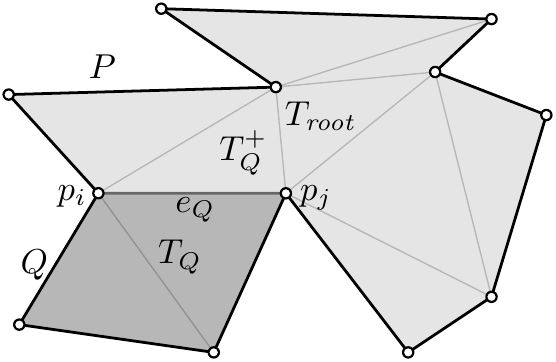}
  \caption{A triangulation with labels for the pocket $Q$ (shaded dark) and triangles $T_Q$ and $T_Q^+$ incident to edge~$e_Q$ of the triangulation.}
  \label{fig:triangulation}
\end{figure}

A \emph{sketch} is an assignment of the vertices of $G$ to simplices (vertices, edges, or triangles) of the triangulation with the property that, if we draw each vertex anywhere on its assigned simplex, then the result is a triangulation-respecting drawing.
We hence interpret a simplex as a closed region of the plane in the remainder of this paper.

\begin{definition}
A \emph{sketch} of the triangulation is a function $\Gamma$ that assigns vertices of $G$ to simplices of $\T$, such that $(i)$ for any vertex $c_i$ of the cycle, $\Gamma(c_i)=p_i$, and $(ii)$ for any two adjacent vertices $u$ and $v$, there exists a triangle of $\T$ that contains both $\Gamma(u)$ and $\Gamma(v)$.
A sketch of a pocket is defined similarly, except that vertices $p_i$ of the polygon that lie outside the pocket do not need $c_i$ assigned to them.
\end{definition}

We show that a sketch exists (for a pocket or a triangulation) if and only if there is a triangulation-respecting drawing (for that pocket or triangulation).
If a pocket admits a sketch, we call a pocket \emph{sketchable}.
If a particular pocket is sketchable, then so are all of its subpockets, since any sketch for a pocket is also a sketch for any of its subpockets.

We present an algorithm that for any sketchable pocket constructs a sketch, and for any other pocket reports that it is not sketchable.
This algorithm recursively constructs particularly well-behaved sketches for child pockets, and combines these sketches into a new well-behaved sketch.
To obtain a sketch for $\T$, we combine the three well-behaved sketches for the three pockets whose lids are the edges of the root triangle $\Tr$ -- assuming that all three pockets are sketchable.

\subparagraph{Well-behaved sketches}
We restrict our attention to \emph{local} sketches for a pocket $Q$, which assign vertices either to simplices in $Q$ or to the triangle $T_Q^+$ just outside $Q$, and \emph{interior} local sketches, which assign vertices to simplices in $Q$ only.
\begin{restatable}[$\star$]{lemma}{lemSketchSimpleOutside}
    If there is a sketch for pocket~$Q$, then there is a local sketch for~$Q$.
    \label{lem:contractedSketch}
\end{restatable}

Generally speaking, it is advantageous for a sketch to place its vertices as far ``to the outside'' as possible, to generate maximum flexibility when combining sketches.
To capture this intuition, we introduce a preorder $\preceq_Q$ on local sketches of a pocket $Q$, defined as $\Gamma\preceq_Q\Gamma'$ iff $\Gamma(v)\cap T_Q^+\subseteq\Gamma'(v)$ for all vertices $v$. 
Intuitively, maximal elements with respect to this preorder maximize for each vertex, the intersection of its assigned simplex with $T_Q^+$.
We call a local sketch $\Gamma$ of $Q$ \emph{well-behaved} if it is maximal with respect to $\preceq_Q$, and \emph{interior well-behaved} if it is maximal among all interior local sketches of $Q$.
A similar preorder and notion of well-behaved can be defined for sketches of the entire triangulation, by replacing~$T_Q^+$ by $\Tr$ in the definition.

\subparagraph{The construction}
We show in Lemma~\ref{lem:greedyPocket} that for any sketchable pocket $Q$, we can construct a specific interior well-behaved sketch $\L_Q$, and a specific well-behaved sketch $\E_Q$.
Before we can present the proof, we first need to define $\L_Q$ and $\E_Q$.

If $Q$ is a trivial pocket, that is, it consists of a single edge~$e_Q$ of $P$, we define
\[\L_Q(v)=\left\{\begin{array}{ll}
  p_i & \text{if $v=c_i$,}\\
  p_j & \text{if $v=c_j$,}\\
  e_Q & \text{otherwise.}\\
\end{array}\right.\]
For non-trivial pockets $Q$ (that do not consist of a single edge), we will define $\L_Q$ differently.
This definition will rely on the definitions of $\E_L$ and $\E_R$ for the child pockets $L$ and $R$ of $Q$ (whose outer triangles $T_L^+$ and $T_R^+$ equal the inner triangle $T_Q$ of $Q$).
Therefore, we postpone the definition of $\L_Q$ for non-trivial pockets until after the definition of $\E_Q$.
Intuitively, $\E_Q$ pushes those vertices which can be placed anywhere on~$e_Q$ out to~$T_Q^+$ if their neighbors allow this and it pushes all remaining vertices as ``far out as possible''.
Formally, $\E_Q$ is defined in terms of $\L_Q$ and will hence be defined if and only if $\L_Q$ is defined:
\[\E_Q(v)=\left\{\begin{array}{ll}
  T_Q^+                & \text{if $e_Q\subseteq\L_Q(v)$ and 
                       $\forall_{(u,v)\in E}\L_Q(u)\cap e_Q\neq\emptyset$,}\\
  \L_Q(v)\cap e_Q & \text{otherwise, if $\L_Q(v)\cap e_Q\neq\emptyset$,}\\
  \L_Q(v)         & \text{otherwise.}\\
\end{array}\right.\]

\noindent
It remains to define $\L_Q$ for non-trivial pockets.
We will define $\L_Q$ only if $\E_L$ and $\E_R$ are defined for both of its child pockets $L$ and $R$.
We attempt to combine $\E_L$ and $\E_R$ into a sketch~$\L_Q$ by taking the more restrictive placement for each vertex; here an assignment to~$T_L^+$ or~$T_R^+$ is interpreted as ``no placement restriction'' (see Figure~\ref{fig:mergeLR}).
Potentially, $\E_L$ and~$\E_R$ restrict the location of a vertex $v$ in such a way that there is no valid placement for $v$.
In such cases, the following definition assigns that vertex to an ``undefined'' location.

\begin{figure}[t]\centering
  \includegraphics{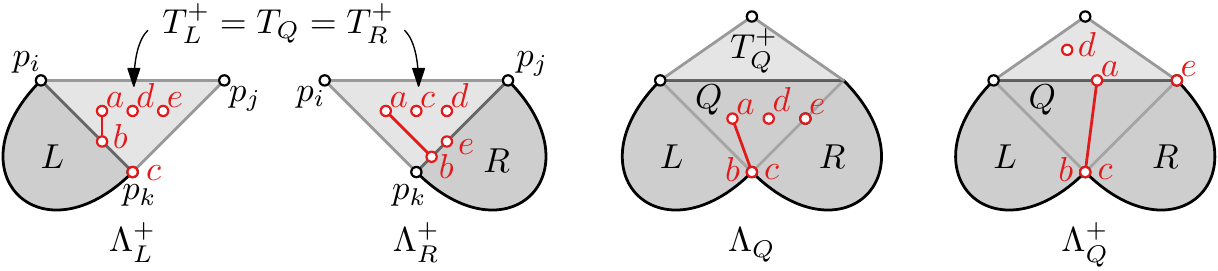}
  \caption{$\E_L$ and~$\E_R$ are merged into~$\L_Q$ which is then transformed into~$\E_Q$ for a subgraph of~$G$. Vertex $c$ is constrained to $p_k$ in $\E_L$ and can lie anywhere in $T_R^+$ in $\E_R$, hence $c$ is constrained to $p_k$ in both $\L_Q$ and $\E_Q$. Vertex $a$ can lie anywhere in $T_L^+ = T_R^+ = T_Q$ in $\L_Q$; since $a$ is connected to $b$ it is pushed to the edge $(p_i,p_j)$ in $\E_Q$ (and not further). Vertex $d$ can also lie anywhere in $T_L^+ = T_R^+ = T_Q$ in $\L_Q$; since it has no further restrictions it is pushed all the way to $T_Q^+$ in~$\E_Q$.}
  \label{fig:mergeLR}
\end{figure}

\[\L_Q(v)=\left\{\begin{array}{ll}
  \E_L(v)\cap\E_R(v) & \text{if $\E_L(v)\cap\E_R(v)\neq\emptyset$,}\\
  \E_L(v)            & \text{otherwise, if $T_Q=\E_R(v)$,}\\
  \E_R(v)            & \text{otherwise, if $T_Q=\E_L(v)$,}\\
  \text{undefined}   & \text{otherwise.}
\end{array}\right.\]

\noindent
If the above equation assigns any vertex to an ``undefined'' location, we say that $\L_Q$ is \emph{undefined}.
In summary, $\L_Q$ is defined if and only if $\E_L$ is defined, $\E_R$ is defined, and the above equation does not assign any vertex to an undefined location.
We inductively show that $\L_Q$ and $\E_Q$ are defined if and only if the pocket $Q$ is sketchable.
Moreover, if they are defined, then $\L_Q$ and $\E_Q$ are interior well-behaved and well-behaved sketches, respectively.
Lemma~\ref{lem:sketchIfDefined} shows that if both $\L_Q$ and $\E_Q$ are defined, then they are sketches. Lemma~\ref{lem:greedyPocket} shows that if pocket $Q$ is sketchable, then $\L_Q$ and $\E_Q$ are both defined and (interior) well-behaved.

We try to construct a sketch $\Delta$ for the root triangle $\Tr$, which (similar to $\Lambda_Q$ for a non-trivial pocket $Q$) combines well-behaved sketches for its child pockets.
Where a non-trivial pocket $Q$ has two child pockets, $\Tr$ has three child pockets $A$, $B$, and $C$ (with $T_A^+=T_B^+=T_C^+=\Tr$).
The equation for $\Delta$ is analogous to that of $\L_Q$; we say that $\Delta$ is defined if and only if all of $\E_A$, $\E_B$, and $\E_C$ are defined, and the following equation does not assign any vertex to an ``undefined'' location.
  \[\Delta(v)=\left\{\begin{array}{ll}
    \E_A(v)\cap\E_B(v)\cap\E_C(v)
             & \text{if $\E_A(v)\cap\E_B(v)\cap\E_C(v)\neq\emptyset$,}\\
    \E_A(v)  & \text{otherwise, if $\E_B(v)=\E_C(v)=\Tr$,}\\
    \E_B(v)  & \text{otherwise, if $\E_A(v)=\E_C(v)=\Tr$,}\\
    \E_C(v)  & \text{otherwise, if $\E_A(v)=\E_B(v)=\Tr$,}\\
      \text{undefined} & \text{otherwise.}
  \end{array}\right.\]
Lemma~\ref{lem:greedyTriangulation} shows that the triangulation $T$ does not admit a sketch if $\Delta$ is undefined.
Otherwise, Lemma~\ref{lem:sketchIfDefined} shows that $\Delta$ is a sketch for the triangulation.

  \begin{lemma}
    The functions~$\L_Q$,~$\E_Q$ and~$\Delta$ are sketches whenever they are defined.\label{lem:sketchIfDefined}
  \end{lemma}
  \begin{proof}
    These functions are sketches if neighboring vertices are assigned to simplices of a common triangle and~$c_k$ is assigned to~$p_k$ for all~$k$; with~$i\leq k\leq j$ in the case of~$\L_Q$ and~$\E_Q$ with~$Q=Q_{(p_i,p_j)}$.
    We prove that the functions are sketches by induction along the definitions of $\L_Q$ and $\E_Q$.

        First consider $\L_Q$ for a trivial pocket $Q$.
        $\L_Q$ assigns all vertices to simplices of the triangle containing the edge $(p_i,p_j)$, and assigns~$\L_Q(c_i)=p_i$ and~$\L_Q(c_j)=p_j$, so~$\L_Q$ is a sketch.
    
    Next, suppose that $\E_Q$ is defined for a pocket $Q$.
    By induction, $\L_Q$ is defined and a sketch.
    If~$\L_Q$ assigns~$c_k$ to~$p_k$, so does~$\E_Q$.
    If~$\E_Q(v)=T^+_Q$, all neighbors of~$v$ are by definition assigned to simplices of~$T^+_Q$.
    If~$\E_Q(v)\neq T^+_Q$, then~$\L_Q(v)\supseteq \E_Q(v)\neq\emptyset$, so if~$\L_Q$ assigns neighboring vertices to simplices of a common triangle, so does~$\E_Q$.
    So~$\E_Q$ is a sketch.

        Next, suppose that $\L_Q$ is defined for a non-trivial pocket $Q$.
    Then~$\E_L$ and~$\E_R$ are sketches for the subpockets~$L$ and~$R$ of~$Q$ with~$T^+_L=T^+_R=T_Q$.
    Furthermore,~$\L_Q$ assigns~$c_k$ to~$p_k$ for all~$i\leq k\leq j$.
    Suppose for a contradiction that~$\L_Q$ assigns two neighboring vertices to simplices that are not simplices of a common triangle.
    Since both~$\E_L$ and~$\E_R$ assign neighboring vertices to simplices of common triangles, there are neighboring vertices~$u$ and~$v$ such that~$\L_Q(u)\not\subseteq\E_L(u)$ and~$\L_Q(v)\not\subseteq\E_R(v)$.
    So by definition of~$\L_Q$, we have~$\E_L(v)\cap\E_R(v)=\emptyset$ and~$\E_L(u)=\E_R(v)=T_Q$.
    Because~$\E_L$ assigns~$u$ and~$v$ to a common triangle, we have~$\E_L(u)\cap\E_L(v)=T_Q\cap\E_L(v)\neq\emptyset$, contradicting that~$\E_L(v)\cap T_Q=\E_L(v)\cap\E_R(v)=\emptyset$.
    Hence~$\L_Q$ is a sketch.
    
    An analogous argument shows that~$\Delta$ is a sketch if for pockets~$A$,~$B$ and~$C$ with~$T^+_A=T^+_B=T^+_C=\Tr$, each of~$\E_A$,~$\E_B$ and~$\E_C$ are sketches, and for all vertices~$v$, we have~$\Tr\subseteq(\E_B(v)\cap\E_C(v))\cup(\E_A(v)\cap\E_C(v))\cup(\E_A(v)\cap\E_B(v))$ or~$\E_A(v)\cap\E_B(v)\cap\E_C(v)\neq\emptyset$.
  \end{proof}
  By the following lemma,~$\L_Q$ and~$\E_Q$ are defined if and only if~$Q$ has a sketch.
  \begin{lemma}
      If a pocket $Q$ is sketchable, then $\L_Q$ is defined and interior well-behaved, and $\E_Q$ is defined and well-behaved.
      \label{lem:greedyPocket}
  \end{lemma}
  \begin{proof}
    We prove this by structural induction along the definitions of $\L_Q$ and $\E_Q$.
    
    For the base case, consider $\L_Q$ for a trivial pocket $Q$ with lid $e_Q=(p_i,p_{i+1})$.
    As $\L_Q$ is defined unconditionally, it is a sketch by Lemma~\ref{lem:sketchIfDefined}.
    Observe that for any interior local sketch $\Gamma$ of $Q$, we have $\Gamma(c_i)=p_i$ and $\Gamma(c_{i+1})=p_{i+1}$.
    For all vertices~$v\notin\{c_i,c_j\}$, we have~$\Gamma(v)\cap e_Q\subseteq\L_Q(v)$, so $\Gamma\preceq_Q\L_Q$.
    That is, $\L_Q$ is interior well-behaved.

        For the inductive step of $\E_Q$, consider a (not necessarily trivial) sketchable pocket~$Q$ with lid $e_Q=(p_i,p_j)$.
        By induction we may assume that $\L_Q$ is defined and interior well-behaved.
        Therefore $\E_Q$ is defined and a sketch (by Lemma~\ref{lem:sketchIfDefined}).
        It remains to show that~$\E_Q$ is well-behaved, so for a contradiction suppose that it is not.
    Then there exists a local sketch $\Gamma$ of $Q$ and a vertex~$v$ for which~$\Gamma(v)\cap T^+_Q\not\subseteq\E_Q(v)$.
    Because~$\L_Q$ does not assign any vertex to simplices outside~$Q$, we have~$\L_Q(v)\subseteq Q$, and by definition of~$\E_Q$, we have~$\L_Q(v)\cap e_Q\subseteq\E_Q(v)\cap T^+_Q$.
    By interior well-behavedness of $\L_Q$ and assumption that~$\Gamma(v)\cap T^+_Q\not\subseteq\E_Q(v)$, we have~$\Gamma(v)\not\subseteq Q$.
    Therefore,~$\Gamma(v)=T^+_Q$ and for all~$(u,v)\in E$ we have~$\Gamma(u)\cap e_Q\neq\emptyset$.
    By interior well-behavedness of $\L_Q$, we have~$\Gamma(v)\cap e_Q=e_Q\subseteq\L_Q(v)$, and for all~$(u,v)\in E$, that~$\emptyset\neq\Gamma(u)\cap e_Q\subseteq \L_Q(u)$ and hence~$\L_Q(u)\cap e_Q\neq\emptyset$.
    So by definition we have~$\E_Q(v)=T^+_Q$, contradicting that~$\Gamma(v)\cap T^+_Q\not\subseteq\E_Q(v)$, so $\E_Q$ is well-behaved.

        For the inductive step of $\L_Q$, suppose that~$Q$ is a non-trivial sketchable pocket with lid~$e_Q=(p_i,p_j)$.
    Since $Q$ is sketchable, so are the child pockets $L$ and $R$ of $Q$ (with~$T^+_L=T^+_R=T_Q$).
    So inductively, we may assume that $\E_L$ and $\E_R$ are well-behaved sketches for~$L$ and~$R$.
    Let~$\Gamma_L$ be the local sketch for $L$ obtained from $\Gamma$ by replacing $\Gamma(v)$ by $T^+_L=T_Q$ whenever $\Gamma(v)\not\subseteq L$.
    Define~$\Gamma_R$ to be the analogous local sketch for $R$. 
    For the sake of contradiction, assume that~$\L_Q$ is not defined.
    Then there is some vertex~$v$ for which~$\E_L(v)\cap\E_R(v)=\emptyset$ and~$T_Q\not\subseteq\E_L(v)\cup\E_R(v)$.
    Because~$T_Q\not\subseteq\E_L(v)\cup\E_R(v)$, we have~$T_Q\not\subseteq\Gamma_L(v)\cup\Gamma_R(v)$, and therefore~$\Gamma(v)\subseteq L$ and~$\Gamma(v)\subseteq R$.
    So~$\Gamma(v)\subseteq L\cap R=p_m$, where~$p_m$ is the vertex of~$T_Q$ not on the lid of~$Q$.
    Since~$\Gamma(v)\neq\emptyset$, well-behavedness of~$\E_L$ and~$\E_R$ tells us that~$p_m\in\E_L(v)$ and~$p_m\in\E_R(v)$.
    But then~$\E_L(v)\cap\E_R(v)\neq\emptyset$, which is a contradiction, so~$\L_Q$ is defined.
        To show that $\L_Q$ is also interior well-behaved, we show that~$\Gamma(v)\cap e_Q\subseteq\L_Q(v)$ for all $v \in V$.
    By well-behavedness of~$\E_L$ and~$\E_R$, we have~$\Gamma(v)\cap e_Q\subseteq\E_L(v)$ and~$\Gamma(v)\cap e_Q\subseteq\E_R(v)$.
    So~$\Gamma(v)\cap e_Q\subseteq \E_L(v)\cap\E_R(v)$, and by definition of~$\L_Q$ we have~$\E_L(v)\cap\E_R(v)\subseteq\L_Q(v)$, and hence~$\Gamma(v)\cap e_Q\subseteq\L_Q(v)$.
    So $\L_Q$ is interior well-behaved.
  \end{proof}

Similarly, we can show for a given $\T$ of a polygon $P$, that if $\T$ has a sketch, then $\Delta$ is defined and a well-behaved sketch.
\begin{restatable}[$\star$]{lemma}{lemGreedyTriangulation}
  If there exists a sketch for a given triangulation $\T$ of a simple polygon~$P$, then~$\Delta$ is defined and well-behaved.
  \label{lem:greedyTriangulation}
\end{restatable}
The following two corollaries summarize that the existence of a sketch is equivalent to the existence of a well-behaved sketch, both for pockets $Q$ and for a complete triangulation $\T$.
\begin{corollary}
  $\L_Q$ and~$\E_Q$ are defined if and only if~$Q$ has a sketch.\label{cor:definedPocket}
\end{corollary}
\begin{corollary}
  Sketch~$\Delta$ is defined if and only if the triangulation $\T$ has a sketch.\label{cor:definedTriangulation}
\end{corollary}

\subparagraph{Computing triangulation-respecting drawings}\label{sec:algorithm}
Any sketch implies a drawing that places vertices anywhere in their assigned simplex.
Conversely, any triangulation-respecting drawing implies a sketch that assigns vertices to the corresponding simplex.
The definition of~$\Delta$ hence directly results in a~$\bigO(t|V||E|)$-time algorithm both to decide the existence of, and to compute a triangulation-respecting drawing.
We can improve this running time to linear.

\begin{restatable}[$\star$]{theorem}{thmAlgorithmComputeSketch}
\label{thm:AlgorithmComputeSketch}
     There is a linear-time algorithm to decide if $(G,C)$ has a triangulation-respecting drawing for a simple polygon $P$ with fixed triangulation~$\mathcal T$; the same algorithm also constructs a drawing if one exists.
\end{restatable}

\section{Planar triangulation-respecting drawings}\label{sec:planar}

We are given the same input as in Section~\ref{sec:triangulation}, namely an instance $(G,C)$ consisting of a graph~$G$ with $n$ vertices and a cycle $C$ with $t$ vertices, and a simple polygon $P$ with $t$ vertices together with an arbitrary triangulation $\T$ of $P$. 
In addition, we assume that the instance~$(G,C)$ is planar, that is, $G$ has a planar drawing $\emb$ with $C$ on the outer face.
Note that $\emb$ does not necessarily map vertices of $C$ to vertices of $P$.

\begin{figure}[t]\centering
    \includegraphics{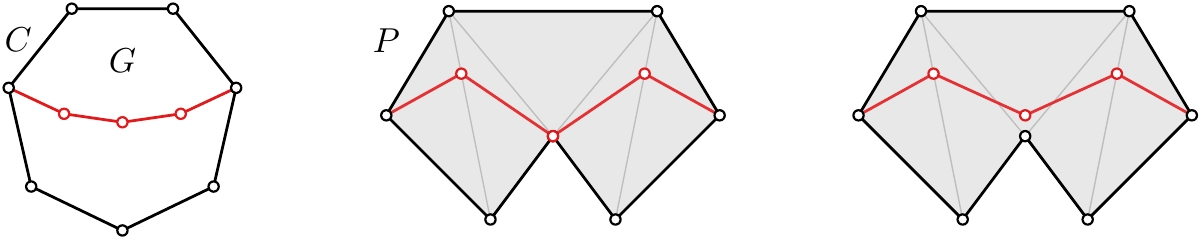}
    \caption{Left: A planar instance. Center: a triangulation-respecting drawing in which two vertices coincide. Right: a perturbed drawing that is planar but not triangulation-respecting.}
    \label{fig:planar}
\end{figure}

Analogously to Section~\ref{sec:triangulation}, we can ask the following question: is there a planar drawing of~$(G,C)$ that respects both $P$ and~$\T$?
However, the answer to this question is often `no', even when both triangulation-respecting drawings and planar polygon-respecting drawings exist.
Consider, for example, Figure~\ref{fig:planar}: a planar triangulation-respecting drawing for this combination of $(G,C)$, $P$, and $\T$ does not exist; any drawing inside $P$ either places two vertices on top of each other, or edges cross edges of the triangulation.
Nonetheless, triangulation-respecting drawings are a useful tool for our final goal of constructing planar polygon-respecting drawings.
For example, the triangulation-respecting drawing of Figure~\ref{fig:planar} can be perturbed infinitesimally to obtain a planar polygon-respecting drawing (that is not triangulation-respecting).
In this section, we show that if a planar instance $(G,C)$ has a triangulation-respecting drawing, then it also has a \emph{weakly-planar triangulation-respecting} one, that is, a triangulation-respecting drawing that is planar and polygon-respecting after infinitesimal perturbation. (Such that vertices are moved to a simplex of $\T$ that contains the original location.)
Hence, the algorithm described in Section~\ref{sec:triangulation} can decide for a planar instance $(G,C)$ whether there is a weakly-planar triangulation-respecting drawing.

Consider now a planar drawing $\emb$ of $(G,C)$. We call the triple $(G,C,\emb)$ a \emph{plane instance}. 
For a weakly-planar triangulation-respecting drawing $\W$, we say that $\W$ \emph{accommodates} $(\emb,\T)$ if there exists a planar polygon-respecting infinitesimal perturbation $\Wt$ of $\W$ that is isotopic to $\emb$ in the plane.
That is, one can be continuously deformed into the other without introducing crossings with itself.
In the following, we construct a weakly-planar triangulation-respecting drawing that accommodates $(\emb,\T)$.
A plane instance $(G,C,\emb)$ is \emph{sketchable} if $(G,C)$ has a sketch (for $\T$).
Recall here, that a sketch does not have a notion of planarity.
However, we show in Theorem~\ref{thm:planarSketchable} that any sketchable plane instance $(G,C,\emb)$ has a drawing $\W$ which accommodates $(\emb,\T)$.

\subparagraph{Minimal plane instances} We show how to transform any plane instance $(G,C,\emb)$ into a \emph{minimal} plane instance, while preserving its sketchability. 
First of all, we carefully triangulate~$(G,C,\emb)$, so as not to influence sketchability (see Lemma~\ref{lem:planar-triangulate} in Appendix~\ref{app:Section4}).
If all faces of~$\emb$ interior to~$C$ are triangles, we call~$(G,C,\emb)$ a \emph{triangulated instance}. A triangulation of a plane instance~$(G,C,\emb)$ is a triangulated instance~$(G',C,\emb')$ such that~$G'$ is a supergraph of~$G$ (with potentially additional vertices) and~$\emb$ is the restriction of~$\emb'$ to~$G$. 
Second, we remove the interior of all separating triangles (see Lemma~\ref{lem:planar-sep-triangle} in Appendix~\ref{app:Section4}).
If $G$ does not have any separating triangles, then we contract any edge not on $C$ that preserves sketchability (see Lemma~\ref{lem:planar-contraction} in Appendix~\ref{app:Section4}) and remove the interior of any separating triangles this edge contraction might create.
If no further simplifications are possible, we call $(G,C,\emb)$ minimal.

\begin{lemma}
  \label{lem:planar-minimal-drawing}
  Every sketchable minimal plane instance~$(G,C,\emb)$ has a drawing that accommodates $(\emb,\T)$.
\end{lemma}
\begin{proof}
    Consider a pocket~$Q=Q_{(p_i,p_j)}$.
    We claim that if~$\E_Q(v)=p_k$ for some~$k\in[i,j]$, then~$v=c_k$, and moreover that if~$\E_Q(v)=e_Q$, then~$Q$ does not consist of a single edge and~$v$ is a neighbor of $c_m$, where~$p_m$ is the third vertex of~$T_Q$.
    If~$Q$ consists of a single edge, then the claim clearly holds.
    If~$Q$ does not consist of a single edge, we have by induction that the claim holds for the subpockets~$L=Q_{(p_i,p_m)}$ and~$R=Q_{(p_m,p_j)}$.
    If~$\E_L(v)=e_L$ for some~$v$, then we claim that the instance is not minimal.  
    Let~$p_l$ be the third vertex of~$T_L$, and without loss of generality assume that~$v$ is the most counter-clockwise neighbor (according to~$\emb$) of~$c_l$ for which~$\E_L(v)=e_L$ (see Figure~\ref{fig:minimal}).
    Then the (triangular) face counter-clockwise of edge~$(c_l,v)$ is a triangle whose third vertex~$u$ does not have~$\E_L(u)=e_L$.
    Since $u$ has $c_l$ and $v$ as neighbors, $\E_L(u)$ is a simplex of $T_L$, but not $e_L$, so by definition of $\E_L$ it is a vertex of $T_L$.
    By induction, $u$ is therefore $c_m$, $c_l$, or $c_i$.
    Because $u$ is a neighbor of $c_l$, $u$ itself is not $c_l$.
    We argue that we can contract an edge while preserving sketchability, contradicting minimality.
    
    First consider the case where $u=c_m$.
    Then $(c_l,c_m)$ divides $G$ into two subgraphs, to the left and to the right of $(c_l,c_m)$.
    We obtain a new sketch by reassigning all vertices of the right subgraph that are placed outside the pocket with lid $(p_l,p_m)$ to that lid.
    By construction, the triangle $c_l,c_m,v$ lies right of $(c_l,c_m)$.
    Then $v$ and its neighbors are reassigned to $p_l$, $p_m$, or $(p_l,p_m)$, so contracting the edge $(u,v)$ maintains sketchability, contradicting minimality.
    
    Now consider the case where $u=c_i$.
    If $\E_Q(u)=\E_Q(v)$, it is clear that we can contract the edge and preserve a sketch.
    So since $\E_Q(v)$ is either $p_i$ or $p_m$, we have $\E_Q(v)=p_m$, but then $\E_R(v)$ is not $T_Q$, so either $\E_R(v)=p_m$ or $\E_R(v) = e_R$.
    The first case implies $v = c_m$ and hence contradicts $\E_L(v) = e_L$.
    In the second case, the inductive hypothesis implies that $v$ is a neighbor of the third vertex $p_r$ of the triangle $T_R$.
    Now consider the subgraph of $G$ that is right of the path consisting of the edges $(p_l,v)$ and $(v,p_r)$.
    Any of its vertices that is assigned outside $L$ and $R$ can be assigned to $p_m$, maintaining a sketch.
    There exists a path from $c_m$ to $v$ that avoids $c_l$ and $c_r$, and the last edge of this path can be contracted, contradicting minimality.
    So~$\E_L(v)\neq e_L$ and symmetrically~$\E_R(v)\neq e_R$.
    By definition,~$\E_Q$ assigns (for~$k\in[i,j]$) only~$c_k$ to~$p_k$, and only neighbors of~$c_m$ to~$e_Q$, so the claim holds.

    Therefore, in the sketch~$\Delta$, all vertices other than those of~$C$ are assigned to edges of~$\Tr$, or~$\Tr$ itself.
    If there is such a vertex not on~$C$, then contracting an edge between~$C$ and~$G\setminus C$ yields an instance with a sketch, contradicting minimality.
    So all vertices in a minimal instance lie on~$C$.
    Because~$G$ is triangulated, this means that its edges coincide with those of the triangulation of~$P$.
    So the instance clearly has an accommodating drawing.
\end{proof}

\begin{figure}[t]
    \centering
    \includegraphics{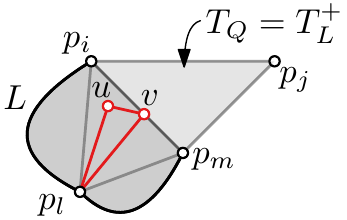}
    \caption{The vertex $u$ must be assigned to $p_i$ or $p_m$. In both cases an edge can be contracted while preserving the existence of a sketch.}
    \label{fig:minimal}
\end{figure}

The proof of Theorem~\ref{thm:planarSketchable} shows that the accommodating drawing of the minimal instance obtained from the simplification procedure (if that instance is sketchable) can be extended to be an accommodating drawing for the original instance, by undoing the simplification steps.

\begin{restatable}[$\star$]{theorem}{thmPlanarSketchable}\label{thm:planarSketchable}
    A plane instance~$(G,C,\emb)$ has a drawing that accommodates $(\emb,\T)$ if and only if $(G,C,\emb)$ is sketchable.
\end{restatable}

The algorithm implied by Theorem~\ref{thm:AlgorithmComputeSketch} can check in linear time if a plane instance $(G,C,\emb)$ has a sketch and via Theorem~\ref{thm:planarSketchable} the same algorithm can decide in linear time if $(G,C,\emb)$ has an accommodating drawing. This drawing can be constructed in polynomial time, following the (polynomial number of) steps in the minimization procedure.

\section{Sufficient conditions for polygon-universality}\label{sec:universality}

In Section~\ref{sec:necessary} we proved that the \ref{cond:banana} and \ref{cond:ninja} Conditions are necessary for an instance $(G,C)$ to be polygon-universal. Here we show, using triangulation-respecting drawings, that these two conditions are sufficient as well. In Sections~\ref{sec:algorithm} and~\ref{sec:planar} we argued that an instance $(G,C)$ has a triangulation-respecting drawing for a triangulation $\T$ of $P$ if and only if it has a sketch for $\T$; we also gave an algorithm that tests whether such a sketch exists. Below we show that if the \ref{cond:banana} and \ref{cond:ninja} Conditions are satisfied for an instance $(G,C)$ then it has a sketch for any triangulation $\T$. We do so by examining the testing algorithm more closely.

We first show that the \ref{cond:banana} Condition alone already implies that each pocket has a sketch. The \ref{cond:ninja} Condition then allows us to combine sketches at the root $\Tr$ of $\T$. Denote by $Q_e=Q_{(p_i,p_j)}$ the pocket with lid $e = (p_i,p_j)$. We want to determine whether an individual vertex can be drawn outside a pocket. Definition~\ref{def:pulled} relates the position of a vertex in a sketch of a pocket to its distance to points on the cycle, see Figure~\ref{fig:pulledCases}.

\begin{figure}[h]
    \centering
  \includegraphics{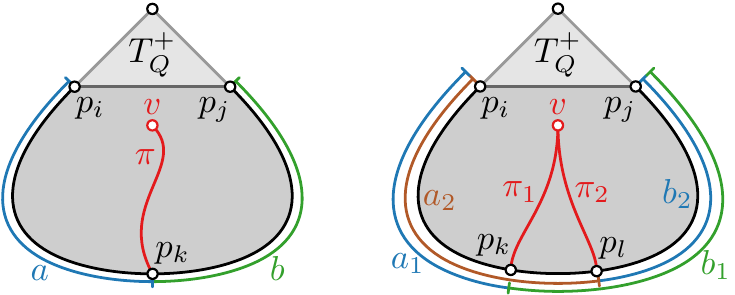}
  \caption{Cases for Definition~\ref{def:pulled}. Left: we have~$\pi\leq a$ and~$\pi\leq b$. Right: we have~$\pi_1\leq a_1+1$, $\pi_1\leq b_1-1$, $\pi_2\leq a_2-1$, and~$\pi_2\leq a_2+1$.}
  \label{fig:pulledCases}
\end{figure}

\begin{definition}\label{def:pulled} 
Vertex~$v$ is \emph{pulled} by pocket~$Q_{(p_i,p_j)}$ if and only if either of the following two conditions hold:
\begin{enumerate}
    \item for some $k\in\{i,\dots,j\}$, we have $d_G(v,c_k) \leq \min(k-i,j-k)$;
  \item for some $k,l\in\{i,\dots,j\}$, we have $d_G(v,c_k) \leq \min(k-i+1,j-k-1)$ and\\
      \hphantom{for some $k,l\in\{i,\dots,j\}$, we have }$d_G(v,c_l)\, \leq \min(l-i-1,j-l+1)$.
\end{enumerate}
\end{definition}

Let $Q$ be a pocket and let $v$ be an arbitrary vertex. Either the well-behaved sketch $\E_Q$ of~$Q$ places $v$ outside of $Q$ or Lemma~\ref{lem:pulledPoint} characterizes where in $Q$ vertex $v$ ``is stuck''. 

\begin{lemma}
  \label{lem:pulledPoint}
    Let $v$ be a vertex and~$Q=Q_{(p_i,p_j)}$ with~$i<j<i+t$ be a sketchable pocket.
    \begin{enumerate}
        \item If~$p_j\notin\E_Q(v)$, then~$v=c_{j-1}$ or there exists a triangle in~$Q$ with vertices~$p_a,p_b,p_j$ and~$i\leq a<b<j$ such that for pocket~$Q'=Q_{(p_a,p_b)}$,~$\E_{Q'}(v)\neq T^+_{Q'}$.
        \item If~$p_i\notin\E_Q(v)$, then~$v=c_{i+1}$ or there exists a triangle in~$Q$ with vertices~$p_i,p_a,p_b$ and~$i<a<b\leq j$ such that for pocket~$Q'=Q_{(p_a,p_b)}$,~$\E_{Q'}(v)\neq T^+_{Q'}$.
    \end{enumerate}
\end{lemma}
  \begin{proof}
    Statements (1) and (2) can be proved using symmetric arguments, so we prove only statement (1).
    We proceed by induction on the size of~$Q$.
    If~$Q$ consists of the single edge~$(p_i,p_j)$, then~$i=j-1$, and if~$p_j\notin\E_Q(v)$, then~$v=c_i=c_{j-1}$.
    So assume that~$Q$ does not consist of a single edge.
    Let~$p_m$ with~$i<m<j$ be the third vertex of~$T_Q$ and let~$L=Q_{(p_i,p_m)}$ and~$R=Q_{(p_m,p_j)}$.
    If~$p_j\notin\E_Q(v)$, then~$\E_L(v)$ or~$\E_R(v)$ does not contain~$p_j$.
    If~$p_j\notin\E_R(v)$, then by induction we are done.
    So assume that~$p_j\in\E_R(v)$ and hence that~$p_j\notin\E_L(v)$.
    This means that~$\E_L(v)\neq T^+_L$, completing the proof.
  \end{proof}

Lemma~\ref{lem:pulledTriangle} and~\ref{lem:doublePulled} relate the characterization in Lemma~\ref{lem:pulledPoint}, which uses the well-behaved sketch $\E_Q$, to the requirements on graph and cycle distances expressed in Definition~\ref{def:pulled}. Lemma~\ref{lem:pulledTriangle} covers the first condition of Definition~\ref{def:pulled}, while Lemma~\ref{lem:doublePulled} covers the remainder.

\begin{lemma}
    Assume that pocket~$Q=Q_{(p_i,p_j)}$ with~$i<j<i+t$ admits a sketch.
    If~$\E_Q(v)\neq T^+_Q$, then~$v$ is pulled by~$Q$.
    \label{lem:pulledTriangle}
\end{lemma}
\begin{proof}
    If~$\E_Q(v)\neq T^+_Q$, then either~$e_Q\not\subseteq\L_Q(v)$ or for some neighbor~$u$ of~$v$,~$\L_Q(u)$ does not intersect $e_Q$.
    We proceed by induction on the size of~$Q$. 
    
    \subparagraph{\boldmath $Q$ is trivial}
    If~$Q$ is trivial, it consists of one edge~$e_Q$, and~$\L_Q(u)$ intersects $e_Q$ for all~$u$.
    By assumption that $\E_Q(v)\neq T^+_Q$, we have $e_Q\not\subseteq\L_Q(v)$ by definition of $\E_Q$, so~$v=c_i$ or $v=c_j$.
    Hence,~$d_G(v,c_k)=0\leq\min\{k-i,j-k\}=0$ for some~$k\in\{i,j\}$, so~$v$ is pulled by~$Q$.

        \subparagraph{\boldmath $Q$ is non-trivial}
        Next, suppose that $Q$ is non-trivial, and thus contains the triangle $T_Q$.
        Let $p_m$ (with~$i<m<j$) be the third vertex of~$T_Q$ and let~$L=Q_{(p_i,p_m)}$ and~$R=Q_{(p_m,p_j)}$ be the two subpockets of~$Q$ with~$T^+_L=T^+_R=T_Q$.
    If~$e_Q\not\subseteq\L_Q(v)$, then~$T_Q\neq\E_L(v)$ or~$T_Q\neq\E_R(v)$, so by induction~$v$ is pulled by~$L$ or~$R$, and hence also by~$Q$.
    So assume that~$e_Q\subseteq\L_Q(v)$ and there exists some neighbor~$u$ of~$v$ for which~$\L_Q(u)$ does not intersect $e_Q$.
    It follows by construction that~$\E_L(v)=\E_R(v)=T_Q$.
    Because~$u$ is assigned to a simplex of the same triangle as~$v$, we have~$\E_L(u)\subseteq T_Q$ and~$\E_R(u)\subseteq T_Q$, so~$\L_Q(u)\subseteq T_Q$.
    Since~$\L_Q(u)$ does not intersect $e_Q$, we have~$\L_Q(u)=p_m$ and hence~$\E_L(u)\cap\E_R(u)=p_m$.
    So either~$\E_L(u)=e_L$ and~$\E_R(u)=e_R$, or~$\E_L(u)$ or~$\E_R(u)$ is~$p_m$.
    We consider these cases separately.

    \subparagraph{Edge case}
    Suppose that~$\E_L(u)=e_L$ and~$\E_R(u)=e_R$.
    Then~$m\in[i+2,j-2]$ and by induction~$u$ is pulled by both~$L$ and~$R$.
    We distinguish three cases depending on what causes~$u$ to be pulled by $L$ and $R$, and show in each case that $v$ is pulled by $Q$.
    \begin{enumerate}
        \item If there exists some~$l\in[i,m]$ with~$d_G(u,c_l)\leq\min(l-i-1,m-l+1)$, then 
            \begin{align*}
                d_G(v,c_l)
                &\leq d_G(u,c_l)+1\\
                &\leq\min(l-i-1,m-l+1)+1\\
                &\leq\min(l-i,j-l)\text{,}
            \end{align*}
            so~$v$ is pulled by~$Q$.
        \item Symmetrically, $v$ is pulled by $Q$ if~$d_G(u,c_k)\leq\min(k-m+1,j-k-1)$ for some~$k\in[m,j]$.
        \item In the remaining case, there exist~$k\in[i,m]$ and~$l\in[m,j]$ with~$d_G(u,c_k)\leq\min(k-i,m-k)$ and~$d_G(u,c_l)\leq\min(l-m,j-l)$.
            Therefore
            \begin{align*}
                d_G(v,c_k)
                &\leq d_G(u,c_k)+1\\
                &\leq\min(k-i,m-k)+1\\
                &\leq\min(k-i+1,m-k+1)\\
                &\leq\min(k-i+1,j-k-1)
                \end{align*}
            and symmetrically~$d_G(v,c_l)\leq\min(l-i-1,j-l+1)$, so~$v$ is pulled by~$Q$.
    \end{enumerate}
    \subparagraph{Corner case}
  Assume that~$\E_L(u)=p_m$ (the case~$\E_R(u)=p_m$ is symmetric).
  Then~$p_i\notin\E_L(u)$, so by Lemma~\ref{lem:pulledPoint}, we have either~$u=c_{i+1}$ or there exists some pocket~$Q'=Q_{(p_a,p_b)}$ with~$i<a<b\leq m$ such that~$\E_{Q'}(u)\neq T^+_{Q'}$.
  If~$u=c_{i+1}$, then for~$k=i+1$ we have~$d_G(v,c_k)\leq d_G(u,c_k)+1\leq 1\leq\min(k-i,j-k)$ in which case~$v$ is pulled by~$Q$.
  Otherwise,~$u$ is by induction pulled by some pocket~$Q_{(p_a,p_b)}$ with~$i<a<b\leq m<j$, and since~$d_G(v,u)\leq 1$, the triangle inequality shows that~$v$ is pulled by~$Q_{(p_i,p_j)}$.

  By induction, $v$ is pulled by $Q$ whenever $\E_Q(v)\neq T^+_Q$.
\end{proof}
\begin{restatable}[$\star$]{lemma}{lemDoublePulled}
  \label{lem:doublePulled}
  Assume that pocket~$Q=Q_{(p_i,p_j)}$ with~$i<j<i+t$ admits a sketch.
  If~$v$ is pulled by~$Q$ but there exists no~$k\in[i,j]$ such that~$d_G(v,c_k)\leq\min(k-i,j-k)$, then there exist~$k,o,l\in[i,j]$ such that~$k<o<l$ and a vertex~$x\neq v$ with~$d_G(x,c_k)\leq\min(k-i,o-k)\leq\min(k-i+1,j-k-1)-d_G(x,v)$ and~$d_G(x,c_l)\leq\min(l-o,j-l)\leq\min(l-i-1,j-l+1)-d_G(x,v)$.
\end{restatable}

Lemma~\ref{lem:arcPocketSketch} ties together our preceding arguments to show that if a pocket has no sketch, then the \ref{cond:banana} Condition is violated. This directly implies Corollary~\ref{cor:bananaPocketSketch}.

\begin{restatable}[$\star$]{lemma}{lemArcPocketSketch}
  \label{lem:arcPocketSketch}
  If pocket~$Q=Q_{(p_i,p_j)}$ with~$i<j<i+t$ has no sketch, then there exist~$i\leq k\leq l\leq j$ such that~$d_G(c_k,c_l)<d_C(c_k,c_l)$.
\end{restatable}

\begin{corollary}
  \label{cor:bananaPocketSketch}
  If the \ref{cond:banana} Condition is satisfied, then all pockets have a sketch.
\end{corollary}

We now established that the \ref{cond:banana} Condition implies that each pocket has a sketch. If additionally the \ref{cond:ninja} Condition is satisfied, then Theorem~\ref{thm:bananaNinjaSketch} shows that we can combine the sketches for the three pockets, whose lids are the edges of the root triangle $\Tr$, to obtain a sketch, and hence a triangulation-respecting drawing, for $(G,C)$.
    
\begin{restatable}[$\star$]{theorem}{thmPolygonUniversal}
  \label{thm:bananaNinjaSketch}
  Let $(G,C)$ be an instance that satisfies the \ref{cond:banana} and \ref{cond:ninja} Conditions.
  Then $(G,C)$ has a triangulation-respecting drawing for any triangulation $\T$ of any simple polygon $P$.
\end{restatable}
  
\begin{corollary}
  Let $(G,C,\emb)$ be a plane instance that satisfies the \ref{cond:banana} and \ref{cond:ninja} Conditions.
  Then $(G,C)$ has a drawing that accommodates $(\emb,\T)$ for any triangulation $\T$ of any simple polygon $P$.
\end{corollary}

\section{Discussion and Conclusion}     
 
We have characterized the (planar) polygon-universal graphs $(G,C)$ by means of simple combinatorial conditions involving (graph-theoretic) distances along the cycle $C$ and in the graph $G$.  In particular, this shows that, even though the recognition of polygon-universal graphs most naturally lies in~$\forall\exists \mathbb R$, it can in fact be tested in polynomial time, by explicitly checking the Pair and the Triple Conditions.
Our main open question concerns the restriction to simple polygons without holes.  Can a similar characterization be achieved in the presence of holes? Or is the polygon-universality problem for simple polygons with holes $\forall\exists \mathbb R$-complete?

Another interesting question concerns the running time for recognizing polygon-universal graphs.  Testing the Pair and Triple Conditions naively requires at least~$\Omega(n^3)$ time.  On the other hand, at least in the non-planar case, given $(G,C)$ and a polygon $P$ with arbitrary triangulation $T$, we can in linear time either find an extension or a violation of the Pair/Triple Condition, which shows that the instance is not polygon-universal. (Recall that a polygon-extension for $P$ might exist, though not one that respects $T$, see Figure~\ref{fig:nonrespecting}). For planar instances, the contraction to minimal instances causes an additional linear factor in the running time. Can (planar) polygon-universality be tested in $o(n^2)$ time?

\bibliography{flex}

\newpage

\appendix
\section{Omitted proofs from Section~\ref{sec:triangulation}}\label{app:Section3}

\lemSketchSimpleOutside*
\begin{proof}
        Let $\Gamma$ be an arbitrary sketch for pocket $Q$.
        To obtain a local sketch from $\Gamma$, we reassign all vertices that are assigned to a simplex that is not contained in $Q$ to the triangle $T^+_Q$ as follows.
        \[\Gamma^*(v):=\left\{\begin{array}{ll}
        T^+_Q      & \text{if $\Gamma(v)\not\subseteq Q$,}\\
        \Gamma(v)  & \text{otherwise.}
      \end{array}\right.\]
      By definition, the function places vertices either inside $Q$, or inside the triangle $T_Q^+$.
      Hence, to show that $\Gamma^*$ is a local sketch, it suffices to show that it is a sketch.
      
        Assume without loss of generality that $Q=Q_{(p_i,p_j)}$.
    For all vertices~$c_k$ with~$i\leq k\leq j$, we have~$\Gamma(c_k)=p_k\subseteq Q$.
    Hence,~$\Gamma^*(c_k)=p_k$ for all~$i\leq k\leq j$.
    To show that~$\Gamma^*$ is a sketch, it remains to show that for each edge~$(u,v)\in E$, the simplices~$\Gamma^*(u)$ and~$\Gamma^*(v)$ are simplices of a common triangle.
    We consider three cases based on whether $\Gamma(u)$ and $\Gamma(v)$ are subsets of $Q$.
    
    \mypar{Both} $\Gamma(u)$ and $\Gamma(v)$ are subsets of $Q$. Then $\Gamma^*(u)=\Gamma(u)$ and $\Gamma^*(v)=\Gamma(v)$, and the property follows since $\Gamma$ is a sketch.
    
    \mypar{Neither} $\Gamma(u)$ and~$\Gamma(v)$ are not subsets of~$Q$. Then~$\Gamma^*(u)=\Gamma^*(v)=T_Q^+$ and the property holds.
    
    \mypar{One} $\Gamma(u)\subseteq Q$ and~$\Gamma(v)\not\subseteq Q$ (the other case is symmetric).
              Then since~$\Gamma(u)$ is a simplex of a triangle containing~$\Gamma(v)$, we have~$\Gamma(u)\subseteq e_Q$.
              So we have~$\Gamma^*(u)\subseteq e_Q$ and~$\Gamma^*(v)=T^+_Q$, and therefore~$\Gamma^*(u)$ and~$\Gamma^*(v)$ are both simplices of~$T^+_Q$.
  \end{proof}
  
  \lemGreedyTriangulation*
    \begin{proof}
    Let~$A$,~$B$, and~$C$ be the pockets whose lid is an edge of~$\Tr$.
    A sketch for the triangulation is automatically a sketch for the pockets~$A$,~$B$, and~$C$.
    So to show that~$\Delta$ is defined, it suffices to show that for each vertex~$v$ we have~$\Tr\subseteq(\E_B(v)\cap\E_C(v))\cup(\E_A(v)\cap\E_C(v))\cup(\E_A(v)\cap\E_B(v))$ or~$\E_A(v)\cap\E_B(v)\cap\E_C(v)\neq\emptyset$.
    Let $\Gamma$ be an arbitrary sketch for the triangulation.
    For well-behavedness, it suffices to show that~$\Gamma(v)\cap\Tr\subseteq\Delta(v)$.
    If~$\Gamma$ assigns~$v$ to a simplex not contained in~$B\cup C$ then it follows from Lemmas~\ref{lem:contractedSketch} and~\ref{lem:greedyPocket} that~$\E_B(v)=\Tr$ and~$\E_C(v)=\Tr$, in which case~$\Tr\subseteq\E_B(v)\cap\E_C(v)\subseteq(\E_B(v)\cap\E_C(v))\cup(\E_A(v)\cap\E_C(v))\cup(\E_A(v)\cap\E_B(v))$ and we are done.
    Similarly, we are done if~$\Gamma$ assigns~$v$ to a simplex not contained in~$A\cup C$ or~$A\cup B$.
    Hence~$\Gamma(v)=\Tr$ or~$\Gamma(v)$ is a vertex of~$\Tr$.
    If~$\Gamma(v)=\Tr$, then also~$\E_A(v)=\E_B(v)=\E_C(v)=\Tr$ since they are well-behaved. So assume that~$\Gamma(v)$ is a vertex of~$\Tr$.
    Let~$p=B\cap C$ be the vertex shared by pockets~$B$ and~$C$ and assume without loss of generality that~$\Gamma(v)=p$ (symmetric arguments apply to the other vertices).
    Then~$\E_A(v)=\Tr$ and~$p\in\E_B(v)$ and~$p\in\E_C(v)$, so~$p\in\E_A(v)\cap\E_B(v)\cap\E_C(v)\neq\emptyset$ and $\Delta$ is defined.
    Now it follows immediately from Lemma~\ref{lem:greedyPocket} that~$\Gamma(v)\cap\Tr\subseteq\Delta(v)$.
  \end{proof}

In the remainder of this section, we provide a linear-time algorithm to decide if $(G,C)$ has a triangulation-respecting drawing if one exists, and constructs such a drawing if so. The proof of Theorem~\ref{thm:AlgorithmComputeSketch} describes the algorithm and also analyses its running time. The rather technical correctness proof is encapsulated in the separate Lemma~\ref{lem:algorithmStep}.

\thmAlgorithmComputeSketch*
\begin{proof}
Observe that~$\Delta(v)$ is defined as an intersection of terms~$\E_Q(v)$ that do not assign~$v$ to~$T^+_Q$.
  So if~$\E_Q(v)\neq T^+_Q$ for some pocket~$Q$, then~$\Delta(v)\subseteq\E_Q(v)$ (if~$\Delta$ is defined).
  Based on this observation, we can assign vertices based on a set of pockets.
  Let~$\mathcal{Q}$ be a set of (non-nested) pockets that all have a sketch.
  We combine the sketches~$\E_Q$ with~$Q\in\mathcal{Q}$ by defining a function~$\E_\mathcal{Q}$ that assigns each vertex to a simplex of the triangulation, the empty set, or the entire polygon~$P$ (to indicate that~$\E_Q(v)=T^+_Q$ for all~$Q\in\mathcal{Q}$).
  \[\E_\mathcal{Q}(v)=\left\{\begin{array}{ll}
    P & \text{if~$\E_Q(v)=T^+_Q$ for all~$Q\in\mathcal{Q}$,}\\
    \bigcap_{Q\in\mathcal{Q}\text{ and }\E_Q(v)\neq T^+_Q}\E_Q(v) & \text{otherwise.}
  \end{array}\right.\]
  If the triangulation has a sketch and~$A$,~$B$ and~$C$ are the three pockets whose lids are edges of~$\Tr$, we can obtain~$\Delta$ from~$\E_{\{A,B,C\}}$ by replacing the value~$P$ by~$\Tr$ in the above definition.
  We compute~$\E_{\{A,B,C\}}$ recursively and maintain a table~$S$ with the property that~$S[v]=\E_{\mathcal{Q}}(v)$ for all~$v\in V$.
  For~$\mathcal{Q}=\emptyset$, we initialize~$S[v]:=\E_\mathcal{Q}(v)=P$ for all~$v\in V$ in~$\bigO(|V|)$ time.
  
  Otherwise we derive~$\E_\mathcal{Q}$ from the values~$S[v]=\E_\mathcal{Q'}(v)$ obtained for a set~$\mathcal{Q'}$ of smaller pockets that all have a sketch, or we deduce that the triangulation has no sketch.
  For this, let~$Q=Q_{(p_i,p_j)}$ be an arbitrary pocket of~$\mathcal{Q}$.
  If~$Q$ consists of a single edge, let~$\mathcal{Q'}=\mathcal{Q}\setminus\{Q\}$.
  Update~$S[c_i]:=S[c_i]\cap\{p_i\}$ and~$S[c_j]:=S[c_j]\cap\{p_j\}$.
  If~$S[c_i]=\emptyset$ or~$S[c_j]=\emptyset$, then the triangulation has no sketch (because if~$\Delta$ were defined, then~$\Delta(c_i)$ or~$\Delta(c_j)$ would be empty).
  Otherwise, we have in constant time updated~$S$ such that~$S[v]=\E_\mathcal{Q}(v)$.

  If instead~$Q$ does not consist of a single edge, we define~$\mathcal{Q'}$ as in Lemma~\ref{lem:algorithmStep} and update~$S$ accordingly.
  For the sake of analysis, we define for a vertex or edge~$x$ of the triangulation~$\mathit{num}(x)$ as the number of vertices~$v$ for which~$\E_\mathcal{Q'}(v)=x$, and let~$\mathit{deg}(p_m)$ be the total degree of vertices for which~$\E_\mathcal{Q'}(v)=p_m$.
  If for each simplex of the triangulation, we store the sets of vertices that~$S$ assigns to it in a doubly-linked list, then we can in~$\bigO(\mathit{deg}(p_m)+\mathit{num}(p_m)+\mathit{num}(e_L)+\mathit{num}(e_R))$ time update both~$S$ and the lists of the triangulation simplices.
  By Lemma~\ref{lem:algorithmStep}, we have~$S[v]=\E_\mathcal{Q}(v)$ if for all~$v$ we have~$S[v]\neq\emptyset$.
  If instead~$S[v]=\emptyset$ for some~$v$, then the triangulation has no sketch.
  
  The above procedure has~$\bigO(n)=\bigO(|V|)$ steps, one for each pocket.
  We claim that this procedure computes~$\E_{\{A,B,C\}}$ in~$\bigO(|V|+|E|)$ time.
  Each vertex is assigned to an edge of the triangulation at most once throughout the procedure.
  Moreover, for an edge~$e$, the term~$\mathit{num}(e)$ contributes to the running time in at most one step.
  Therefore, the total contribution of~$\mathit{num}(e)$ over all edges~$e$ is~$\bigO(|V|)$ time.
  Similarly, a vertex is assigned to a vertex of the triangulation at most once throughout the procedure.
  Likewise, for each vertex~$p_k$, the terms~$\mathit{deg}(p_k)$ and~$\mathit{num}(p_k)$ contribute to the running time in at most one step (for totals of respectively~$\bigO(|E|)$ and~$\bigO(|V|)$ time throughout the procedure).
  Therefore, the total procedure takes~$\bigO(|V|+|E|)$ time and we can decide in linear time whether there exists a triangulation-respecting drawing, and output one if so.
\end{proof}

\begin{restatable}{lemma}{lemAlgorithmStep}
    \label{lem:algorithmStep}
    Let~$\mathcal{Q}$ be a set of non-nested pockets and assume that pocket~$Q=Q_{(p_i,p_j)}\in\mathcal{Q}$ does not consist of a single edge.
    Let~$p_m$ with~$i<m<j$ be the third vertex of~$T_Q$, and let~$L=Q_{(p_i,p_m)}$ and~$R=Q_{(p_m,p_j)}$ be subpockets of~$Q$.
    Assume that all pockets in~$\mathcal{Q'}=\{L,R\}\cup\mathcal{Q}\setminus\{Q\}$ have sketches and that~$S'[v]=\E_\mathcal{Q'}(v)\neq\emptyset$ for all~$v$.
    Define
    \[S[v]=\left\{\begin{array}{ll}
      e_Q\cap S'[v] & \text{if~$S'[v]=e_L$ or~$S'[v]=e_R$,}\\
      e_Q\cap S'[v] & \text{if~$S'[v]\not\subseteq L\cup R$ and there is a neighbor~$u$ of~$v$ with~$S'[u]=p_m$,}\\
      S'[v] & \text{otherwise.}
    \end{array}\right.\]
    Then~$Q$ has a sketch and~$S[v]=\E_\mathcal{Q}(v)$ for all~$v$.
  \end{restatable}
  
    \begin{proof}
    We first show that~$\E_Q$ is defined.
    If~$\E_Q$ is not defined, then there is some vertex~$v$ with~$\E_L(v)\cap\E_R(v)=\emptyset$ and~$\E_L(v)\neq T^+_L$ and~$\E_R(v)\neq T^+_R$.
    But then~$S'[v]=\emptyset$, contradicting that~$\E_\mathcal{Q'}(v)\neq\emptyset$, so~$\E_Q$ is defined and therefore~$Q$ has a sketch.

    Suppose that~$S'[v]=e_L$.
    Then we claim that~$\E_{Q'}(v)=T^+_{Q'}$ for all~$Q'\in\mathcal{Q'}\setminus\{L\}$.
    Indeed, otherwise~$S'[v]\subseteq Q'$, but because pockets of~$\mathcal{Q'}$ are not nested, we have~$e_L\not\subseteq Q'$, contradicting that~$S'[v]=e_L$.
    So~$\E_L(v)=e_L$ and~$\E_R(v)=T^+_R$, and therefore~$\E_Q(v)=p_i=e_Q\cap S'[v]$.
    Since all pockets~$Q'\in\mathcal{Q}\setminus\{Q\}$ have~$\E_{Q'}(v)=T^+_{Q'}$, we have~$\E_\mathcal{Q}(v)=\E_Q(v)=e_Q\cap S'[v]$.
    Therefore~$S[v]$ is set correctly if~$S'[v]=e_L$ (or~$S'[v]=e_R$ by a symmetric argument).

    If~$S'[v]\not\subseteq L\cup R$, then~$\E_L(v)=T^+_L$ and~$\E_R(v)=T^+_R$, so~$\L_Q(v)=T_Q$.
    If additionally there is a neighbor~$u$ of~$v$ with~$S'[u]=p_m$, then~$\E_Q(v)=e_Q$ and therefore~$\E_\mathcal{Q}(v)=e_Q\cap\E_{\mathcal{Q}\setminus\{Q\}}(v)=e_Q\cap\E_{\mathcal{Q'}\setminus\{L,R\}}(v)=e_Q\cap\E_\mathcal{Q'}(v)=e_Q\cap S'[v]$.
    Because the second case of the definition of~$S[v]$ applies, we correctly assign~$S[v]=e_Q\cap S'[v]$.
    If instead there is no such neighbor~$u$, then~$\E_Q(v)=T^+_Q$, so~$\E_\mathcal{Q}(v)=\E_{\mathcal{Q}\setminus\{Q\}}(v)=\E_{\mathcal{Q'}\setminus\{L,R\}}(v)=\E_\mathcal{Q'}(v)$.
    Because the third case of the definition of~$S[v]$ applies we correctly assign~$S[v]=S'[v]$.

    Finally, assume that~$S'[v]\neq e_L$ and~$S'[v]\neq e_R$ and~$S'[v]\subseteq L\cup R$, so the third case of the definition of~$S[v]$ applies.
    If there exists some pocket~$Q'\in\mathcal{Q'}\setminus\{L,R\}$ with~$\E_{Q'}(v)\neq T^+_{Q'}$, then~$S'[v]=p_i$ or~$S'[v]=p_j$.
    Consider the case where~$S'[v]=p_i$ (the other case is symmetric).
    Then~$p_i\in\E_L(v)$ and~$\E_R(v)=T^+_R$, in which case~$p_i\in\E_Q(v)$ and hence~$\E_\mathcal{Q}(v)=p_i=\E_\mathcal{Q'}(v)$, so~$S[v]$ is correctly set to~$S'[v]$.
    So assume that all pockets~$Q'\in\mathcal{Q'}\setminus\{L,R\}$ have~$\E_{Q'}(v)=T^+_{Q'}$.
    Then~$\E_{\mathcal{Q'}\setminus\{L,R\}}(v)=P$, so~$\E_\mathcal{Q'}(v)=\E_{\{L,R\}}(v)$ and~$\E_\mathcal{Q}(v)=\E_{\{Q\}}(v)$.
    Therefore,~$S[v]$ is correctly set to~$S'[v]$ if~$\E_{\{Q\}}(v)=\E_{\{L,R\}}(v)$.
    Because~$S'[v]\subseteq L\cup R$, we have~$\E_{\{L,R\}}(v)\neq P$ and therefore it follows that~$\E_{\{L,R\}}(v)=\L_Q(v)\subseteq L\cup R$ and in particular~$e_Q\not\subseteq\L_Q(v)$.
    If moreover~$\L_Q(v)\cap e_Q=\emptyset$, we have~$\E_Q(v)=\L_Q(v)=S'[v]\neq T^+_Q$ and we have correctly set~$S[v]=S'[v]$.
    If instead~$\L_Q(v)\cap e_Q\neq\emptyset$, we have~$\E_Q(v)=\L_Q(v)\cap e_Q$.
    In that case, because~$\L_Q(v)\subseteq L\cup R$ and~$\L_Q(v)\neq e_L$ and~$\L_Q(v)\neq e_R$, we have~$\L_Q(v)=p_i$ or~$\L_Q(v)=p_j$, so~$\E_Q(v)=\L_Q(v)\cap e_Q=\L_Q(v)$ and we have correctly set~$S[v]=S'[v]$.
  \end{proof}

\section{Omitted Proofs from Section~\ref{sec:planar}}\label{app:Section4}

\begin{restatable}{lemma}{lemPlanarTriangulate}
\label{lem:planar-triangulate}
    For any plane instance~$(G,C,\emb)$ that has a sketch for triangulation $\T$, there exists a triangulated instance $(G',C,\emb')$ of $(G,C,\emb)$ that also has a sketch for $\T$.
\end{restatable}

\begin{proof}
    If~$G$ is not connected, we connect it by iteratively applying the following procedure.
    Let~$H$ be the component of~$G$ containing~$C$ and let~$F$ be a face of~$\emb$ whose boundary contains a vertex~$u$ of~$H$ and a vertex of a component~$H'$ of~$G\setminus H$, and connect those vertices by an edge to obtain a new instance~$(G',C,\emb')$ with fewer components, where~$\emb'$ is a copy of~$\emb$ that additionally routes the new edge through~$F$.
    We can transform a sketch of~$G$ into a sketch of~$G'$ by reassigning the vertices of~$H'$ to the simplex that~$u$ is assigned to.

    So assume without loss of generality that~$G$ is connected.
    If~$G$ is not triangulated, we iteratively pick a non-triangular interior face~$F$ of~$\emb$, and consider the cycle~$\{v_1,\dots,v_k\}$ of vertices on the boundary of~$F$.
    We place a copy~$C'=\{v'_1,\dots,v'_k\}$ of this cycle interior to~$F$ and connect~$v_i$ to~$v'_i$ and~$v'_{i+1}$ (where~$v'_{k+1}=v'_1$ and~$v_{k+1}=v_1$), creating~$2k$ triangular faces with vertices~$v_i,v'_i,v'_{i+1}$ and vertices~$v_i,v'_{i+1},v_{i+1}$, and one face with the~$k$ vertices of~$C'$ as boundary.
    The resulting graph has a sketch (simply assign~$v'_i$ to the simplex that~$v_i$ is assigned to), and the restriction of its embedding to~$G$ is~$\emb$.

    We will triangulate~$C'$ as follows by iteratively adding chords to non-triangular faces interior to~$C'$.
    For this, let~$F'$ be a non-triangular face interior to~$C'$ (possibly after adding chords), and let~$C''=\{u_1,\dots,u_l\}$ (with~$l\geq 4$) be the cycle of vertices on its boundary.
    Then the distance (in the graph) between any pair~$u_i$ and~$u_j$ of vertices that are not adjacent on~$C''$ is at least~$2$.
    We will pick a chord of~$C''$ in such a way the resulting graph has a sketch.
    If~$\Delta$ assigns a pair of non-adjacent vertices of~$C''$ to simplices that intersect a common triangle of~$\mathcal{T}$, then we can connect them by an edge to obtain a graph that has a sketch (as well as a planar embedding whose restriction to the original graph is~$\emb$).
    To show that a pair of such non-adjacent vertices exist, consider the vertices~$u_1$ and~$u_3$.
    If~$\Delta(u_1)$ and~$\Delta(u_3)$ intersect a common triangle of~$\mathcal{T}$, then we are done.
    Otherwise, there is an edge~$e_Q$ of the triangulation~$\mathcal{T}$ such that~$\Delta(u_1)$ and~$\Delta(u_3)$ lie on different sides of~$e_Q$.
    But then~$\Delta(u_2)$ must intersect~$e_Q$, and so must~$\Delta(u_i)$ for some vertex~$u_i$ with~$4\leq i\leq l$, in which case~$\Delta(u_2)$ and~$\Delta(u_i)$ both intersect simplices of~$T_Q$.

    This procedure transforms the original planar instance~$(G,C,\emb)$ into a sketchable planar triangulation~$(G',C,\emb')$ of~$(G,C,\emb)$.
  \end{proof}

  \begin{restatable}{lemma}{lemPlanarSepTriangle}
  \label{lem:planar-sep-triangle}
    Suppose that~$(G,C,\emb)$ has a separating triangle and let~$H$ be the subgraph of~$G$ that~$\emb$ draws in the interior of that triangle.
    Then~$G$ has an accommodating drawing for~$(\emb,\T)$ if~$G\setminus H$ has an accommodating drawing for~$(\emb',\T)$ (where~$\emb'$ is the restriction of~$\emb$ to~$G\setminus H$).
  \end{restatable}
  
  \begin{proof}
        Consider an accommodating drawing $\W'$ for $(\emb',\T)$.
    Then $\W'$ places the vertices of the separating triangle that contained~$H$ on a common triangle of~$\T$.
    Let $\W$ be a drawing of $(G,C)$, obtained from $\W'$ by placing place all vertices of $H$ at the same location as one vertex of the separating triangle.
    Specifically, let $\W(v)=\W'(v)$ for all vertices $v$ of $G\setminus H$, and $\W(v)=\W'(s)$ for all vertices $v$ of $H$, where $s$ is a vertex of the separating triangle.
    The perturbation of $\W'$ maps the separating triangle to an empty triangle, and $\W$ can perturb $H$ into its interior in such a way that the resulting drawing is planar.
    Therefore, $\W$ is an accommodating drawing for $(\emb,\T)$.
  \end{proof}
  
  \begin{restatable}[$\star$]{lemma}{lemPlanarContract}
    \label{lem:planar-contraction}
    Let~$(G,C,\emb)$ be a sketchable triangulated instance without separating triangles.
    Suppose that contracting some edge $(u,v)\in E$ not on $C$ yields a graph $G'$ with an accommodating drawing $\W'$ for $(\emb',\T)$, where~$\emb'$ is obtained from~$\emb$ by routing all edges that were incident to~$v$ (except those that connect to a neighbor of $u$) parallel to $e$, and connecting them to the contracted vertex, which $\emb'$ draws at the location of $u$ in $\emb$ (see Figure~\ref{fig:contraction}).
    Then~$(G,C)$ has an accommodating drawing $\W$ for $(\emb,\T)$, where $\W(w)=\W'(w)$ for all $w\notin\{u,v\}$, and $\W(u)=\W(v)$ is the location of the vertex resulting from the contraction.
  \end{restatable}
    \begin{figure}[h]\centering
      \includegraphics{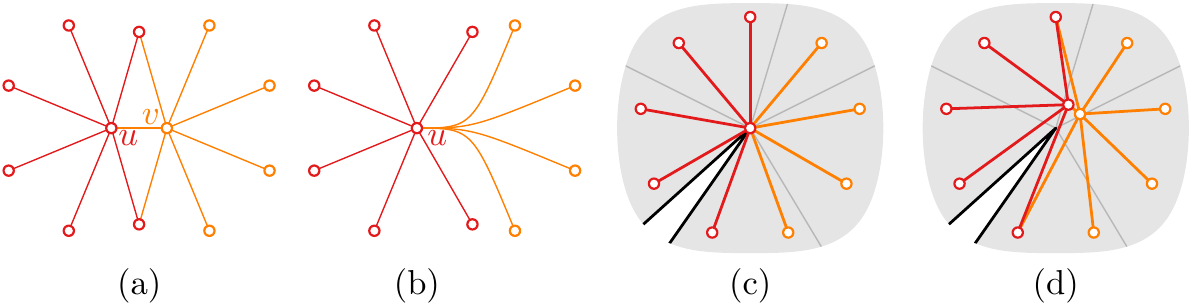}
      \caption{(a) $\emb$ drawn for a neighborhood of~$(u,v)$. (b) $\emb'$ drawn for a neighborhood of the contraction. (c) An accommodating drawing $\W'$ for~$(\emb',\T)$. (d) A perturbation of $\W$.}
      \label{fig:contraction}
    \end{figure}
  \begin{proof}
    It is well-known that contracting an edge $e$ of a triangulation yields a triangulation if and only if $e$ is not part of a separating triangle.
    Hence, $(G',C,\emb')$ is a triangulated instance.
    Denote by $z$ the vertex that is obtained by contracting $(u,v)$, and let $x,y$ be the two vertices of $G$ that (together with $u$ and $v$) form the triangles incident to $(u,v)$ in $G$.
    To see that $\W$ is an accommodating drawing of $(G,C)$ for $(\emb,\T)$, consider the planar polygon-respecting drawing $\Wt'$ obtained by perturbing $\W'$.
    We can obtain a drawing isotopic to $\emb$ as follows: erase an arbitrarily small neighborhood of the vertex $z$ in $\Wt'$ and place vertices $u$ and $v$ in this neighborhood, connected by the edge $(u,v)$.
    Inside the erased neighborhood, extend all erased edges to the vertex $u$ or $v$, depending on which they correspond to.
    Here, we take special care of the edges $(z,x)$ and $(z,y)$, because $(u,x)$, $(v,x)$, $(u,y)$, and $(v,y)$ all are edges of $G$.
    For this, we place two copies of $(z,x)$ and $(z,y)$ side-by-side, which is possible without violating planarity.
    Finally, observe that since we can take the neighborhoods arbitrarily small, we can draw each of the redrawn or extended edges as straight lines.
    The result is a planar polygon-respecting drawing isotopic to $\emb$, and such a drawing can be obtained from $\W$ by arbitrarily small perturbation (by using a sufficiently small neighborhood).
  \end{proof}

    \thmPlanarSketchable*
    \begin{proof}
    Since any accommodating drawing can be transformed into a sketch, there is no accommodating drawing if~$(G,C)$ has no sketch.
    So assume that~$(G,C)$ does have a sketch.
    Lemmas~\ref{lem:planar-triangulate}, \ref{lem:planar-sep-triangle} and~\ref{lem:planar-contraction} show how to transform~$(G,C,\emb)$ into a minimal instance that has a sketch.
    Moreover, these lemmas show that~$(G,C,\emb)$ has an accommodating drawing if the transformed instance~$(G',C,\emb')$ has one.
    Since the transformed instance is minimal and has a sketch, it indeed has an accommodating drawing by Lemma~\ref{lem:planar-minimal-drawing}.
    So~$(G,C,\emb)$ has an accommodating drawing if and only if~$(G,C)$ has a sketch.
  \end{proof}

\section{Omitted Proofs from Section~\ref{sec:universality}}

\lemDoublePulled*
  \begin{proof}
    We prove this by induction on the size of~$Q$.
    Because~$v$ is pulled due to the second case of Definition~\ref{def:pulled},~$Q$ does not consist of a single edge.
    Let~$p_m$ (with~$i<m<j$) be the third vertex of~$T_Q$ and consider the subpockets~$L=Q_{(p_i,p_m)}$ and~$R=Q_{(p_m,p_j)}$ of~$Q$.
    If~$v$ is pulled by a subpocket of~$Q$, then~$v$ is pulled by~$Q$ due to the first case of Definition~\ref{def:pulled}.
    So assume that~$v$ is not pulled by any subpocket of~$Q$.
    The proof of Lemma~\ref{lem:pulledTriangle} tells us that there is a neighbor~$u$ of~$v$ such that either (1)~$u$ is pulled by both~$L$ and~$R$ due to the first case of Definition~\ref{def:pulled}, or (2)~$u$ is pulled by some pocket~$Q_{(p_a,p_b)}$ with~$i<a<b<j$.

    First consider case (1).
    If~$d_G(u,c_m)=0$, then~$d_G(v,p_m)\leq\min(m-i,j-m)$, contradicting that~$v$ is not pulled due to the first case of Definition~\ref{def:pulled}.
    So~$d_G(u,c_m)>0$, and it follows from Definition~\ref{def:pulled} that there exist~$k$ and~$l$ with~$i\leq k<m<l\leq j$ such that~$d_G(u,c_k)\leq\min(k-i,m-k)\leq\min(k-i+1,j-k-1)-d_G(u,v)$ and~$d_G(u,c_l)\leq\min(m-i,j-l)\leq\min(l-i-1,j-l+1)-d_G(u,v)$.
    This completes the proof for case (1).

    If in case (2)~$u$ is pulled by~$Q_{(p_a,p_b)}$ due to the first case of Definition~\ref{def:pulled}, then~$v$ is pulled by~$Q$ due to the first case of Definition~\ref{def:pulled}.
    Hence~$u$ is pulled by~$Q_{(p_a,p_b)}$, but not by the first case of Definition~\ref{def:pulled}.
    So by induction, there exist~$k,o,l\in[a,b]$ such that~$k<o<l$ and a vertex~$x$ with~$d_G(x,c_k)\leq\min(k-a,o-k)\leq\min(k-a+1,b-k-1)-d_G(x,u)$ and~$d_G(x,c_l)\leq\min(l-o,b-l)\leq\min(l-a-1,b-l+1)-d_G(x,u)$.
    Thus, since~$d_G(u,v)\leq 1$ and~$i<a\leq k<o<l\leq b<j$ we have~$d_G(x,c_k)\leq\min(k-i,o-k)\leq\min(k-i+1,j-k-1)-d_G(x,v)$ and~$d_G(x,c_l)\leq\min(l-o,j-l)\leq\min(l-i-1,j-l+1)-d_G(x,v)$.
    Because~$x$ is pulled due to the first case of Definition~\ref{def:pulled} we have~$x\neq v$, completing the proof.
  \end{proof}

\lemArcPocketSketch*
  \begin{proof}
    We prove this by induction on the size of~$Q$.
    If~$Q$ consists of a single edge, then there always exists a sketch and we are done.
    So assume that~$Q$ is not a single edge and let~$p_m$ (with~$i<m<j$) be the third vertex of~$T_Q$, the other two being~$p_i$ and~$p_j$.
    If~$Q$ has no sketch, then by Corollary~\ref{cor:definedPocket} either~$L=Q_{(p_i,p_m)}$ or~$R=Q_{(p_m,p_j)}$ has no sketch, or there exists some vertex~$v$ such that~$\E_L(v)\cap\E_R(v)=\emptyset$ and~$T_Q\not\subseteq\E_L(v)\cup\E_R(v)$.
    If~$L$ or~$R$ has no sketch, then by induction there exist~$k$ and~$l$ such that~$i\leq k\leq l\leq m\leq j$ or~$i\leq m\leq k\leq l\leq j$ and~$d_G(c_k,c_l)<d_C(c_k,c_l)$, in which case we are done.
    So assume that~$L$ and~$R$ both have a sketch, and hence that~$\E_L$ and~$\E_R$ are defined.

    Since~$Q$ has no sketch, there exists some vertex~$v$ such that~$\E_L(v)\cap\E_R(v)=\emptyset$ and~$T_Q\not\subseteq\E_L(v)\cup\E_R(v)$.
    Because~$T_Q\not\subseteq\E_L(v)\cup\E_R(v)$, Lemma~\ref{lem:pulledTriangle} states that~$v$ is pulled by both~$L$ and~$R$.
    Moreover, as~$\E_L(v)\cap\E_R(v)=\emptyset$, we have~$p_m\notin\E_L(v)$ or~$p_m\notin\E_R(v)$, so Lemma~\ref{lem:pulledPoint} applies to~$L$ or~$R$.
    We consider only the case where~$p_m\notin\E_L(v)$, the other case is symmetric.

    We show that there exists some~$k$ with~$i\leq k<m$ such that (1)~$d_G(v,c_k)\leq\min(m-k-1,k-i)$ or (2)~$d_G(v,c_k)\leq\min(m-k-2,k-i+1)$.
    Since~$p_m\notin\E_L(v)$, Lemma~\ref{lem:pulledPoint} states that~$v=c_{m-1}$ or there exists a triangle in~$L$ with vertices~$p_a,p_b,p_m$ and~$i\leq a<b<m$ such that~$v$ is pulled by~$Q_{(p_a,p_b)}$.
    If~$v=c_{m-1}$, then case~(1) applies because for~$k=m-1$ we have~$d_G(v,c_k)=0=m-k-1$ and since~$m-1\geq i$ also~$d_G(v,c_k)\leq k-i$.
    Otherwise, there exists some triangle with vertices~$p_a,p_b,p_m$ and~$i\leq a<b<m$ such that~$v$ is pulled by~$Q_{(p_a,p_b)}$.
    Definition~\ref{def:pulled} applied to~$Q_{(p_a,p_b)}$ tells us that there exists some~$k$ with~$i\leq a\leq k\leq b<m$ such that either~$d_G(v,c_k)\leq k-a\leq k-i$ and~$d_G(v,c_k)\leq b-k\leq m-k-1$ (so case~(1) applies), or~$d_G(v,c_k)\leq k-a+1\leq k-i+1$ and~$d_G(v,c_k)\leq b-k-1\leq m-k-2$ (so case~(2) applies).
    So case~(1) or case~(2) applies.

    Applying Definition~\ref{def:pulled} to~$R$, we also have either (a)~$k'$ with~$m\leq k'\leq j$ such that~$d_G(v,c_{k'})\leq\min(k'-m,j-k')$ or (b)~$k',l'\in[m,j]$ such that~$d_G(v,c_{k'})\leq\min(k'-m+1,j-k'-1)$ and~$d_G(v,c_{l'})\leq\min(l'-m-1,j-l'+1)$.

    Since~$Q$ is a pocket, there is at least one vertex outside~$Q$, so we have~$j-i\leq n-2$.
    For any~$i\leq k\leq l\leq j$, we have~$d_C(c_k,c_l)=\min(l-k,k+n-l)$.
    We show that there exist~$i\leq k\leq l\leq j$ such that~$d_G(c_k,c_l)<d_C(c_k,c_l)$ using the following cases.
    \begin{description}
      \item[(1)(a)]~$d_G(c_k,c_{k'})\leq d_G(v,c_k)+d_G(v,c_{k'})\leq\min(m-k-1,k-i)  +\min(k'-m,j-k')    \leq\min(m-k-1+k'-m,k-i+j-k')  \leq\min(k'-k-1,k+n-2-k')<\min(k'-k,k+n-k')$.
      \item[(1)(b)]~$d_G(c_k,c_{l'})\leq d_G(v,c_k)+d_G(v,c_{l'})\leq\min(m-k-1,k-i)  +\min(l'-m-1,j-l'+1)\leq\min(m-k-2+l'-m,k-i+j-l'+1)\leq\min(l'-k-2,k+n-1-l')<\min(l'-k,k+n-l')$.
      \item[(2)(a)]~$d_G(c_k,c_{k'})\leq d_G(v,c_k)+d_G(v,c_{k'})\leq\min(m-k-2,k-i+1)+\min(k'-m,j-k')    \leq\min(m-k-2+k'-m,k-i+1+j-k')\leq\min(k'-k-2,k+n-1-k')<\min(k'-k,k+n-k')$.
      \item[(2)(b)]~$d_G(c_k,c_{k'})\leq d_G(v,c_k)+d_G(v,c_{k'})\leq\min(m-k-2,k-i+1)+\min(k'-m+1,j-k'-1)\leq\min(m-k-1+k'-m,k-i+j-k')  \leq\min(k'-k-1,k+n-2-k')<\min(k'-k,k+n-k')$.
    \end{description}
    In each case, there exist~$i\leq k\leq l\leq j$ with~$d_G(c_k,c_l)<\min(l-k,k+n-l)=d_C(c_k,c_l)$.
  \end{proof}

\thmPolygonUniversal*

  \begin{proof}
    Let~$\Tr$ be an ear of the triangulation with vertices~$p_i,p_j,p_t$ and~$i\leq j=i+n-2\leq t=i+n-1$.
    Define the pockets~$A=Q_{(p_i,p_j)}$,~$B=Q_{(p_j,p_t)}$ and~$C=Q_{(p_t,p_i)}$.
    Suppose for a contradiction that the triangulation does not have a sketch but the \ref{cond:banana} and \ref{cond:ninja} Conditions are both satisfied.
    By Corollary~\ref{cor:bananaPocketSketch}, each of~$A$,~$B$ and~$C$ have a sketch.
    Moreover by Corollary~\ref{cor:definedTriangulation}~$\Delta$ is not defined.
    So there is some vertex~$v$ such that~$\E_A(v)\cap\E_B(v)\cap\E_C(v)=\emptyset$ and~$\Tr$ is contained in at most one of~$\E_A(v)$,~$\E_B(v)$, and~$\E_C(v)$.
    We have~$v\in\{c_i,c_j,c_t\}$ because otherwise both~$\E_B(v)$ and~$\E_C(v)$ contain~$\Tr$.
    We also have~$\E_A(c_i)=\E_C(c_i)=p_i$ and~$\E_B(c_i)=\Tr$, so~$p_i\in\E_A(c_i)\cap\E_B(c_i)\cap\E_C(c_i)\neq\emptyset$.
    Similarly~$p_j\in\E_A(c_j)\cap\E_B(c_j)\cap\E_C(c_j)\neq\emptyset$.
    So~$v=c_t$ and since~$\E_B(c_t)=\E_C(c_t)=p_t$, at most~$\E_A(c_t)$ contains~$\Tr$.
    But~$\E_A(c_t)\not\subseteq\Tr$, since otherwise~$p_t\in\E_A(c_i)\cap\E_B(c_i)\cap\E_C(c_i)\neq\emptyset$.
    So~$c_t$ is pulled by~$A$.

    Suppose that~$c_t$ is pulled by~$A$ due to the first case of Definition~\ref{def:pulled}.
    Then there is some~$k\in[i,j]$ with~$d_G(c_t,c_k)\leq\min(k-i,j-k)<d_C(c_t,c_k)$, violating the \ref{cond:banana} Condition.
    Similarly, the \ref{cond:banana} Condition is violated if~$c_t$ is additionally pulled by a subpocket~$Q_{(p_a,p_b)}$ of~$A$ with~$i\leq a<b\leq j$ (and~$b-a<j-i$).
    Therefore Lemma~\ref{lem:doublePulled} states that there exists some vertex~$x\neq c_t$ and~$i\leq k<o<l\leq j$ such that~$d_G(x,c_k)\leq\min(k-i,o-k)$ and~$d_G(x,c_l)\leq\min(l-o,j-l)$ and~$d_G(c_k,c_t)\leq d_G(x,c_k)+d_G(x,c_t)\leq\min(k-i+1,j-k-1)$ and~$d_G(c_l,c_t)\leq d_G(x,c_l)+d_G(x,c_t)\leq\min(l-i-1,j-l+1)$.
    The \ref{cond:banana} Condition is violated if~$k\geq t-n/2$ or~$l\leq t-n/2$, so~$k<t-n/2<l$.
    Therefore,~$d_C(c_k,c_t)=k-i+1$ and~$d_C(c_t,c_l)=j-l+1$.
    Since~$k-i+1\leq d_C(c_k,c_l)\leq d_G(c_k,c_t)\leq d_G(x,c_k)+d_G(x,c_t)\leq k-i+1$, we have~$d_G(c_k,c_t)=d_G(x,c_k)+d_G(x,c_t)=k-i+1=d_C(c_k,c_t)$ and similarly~$d_G(c_t,c_l)=d_G(x,c_l)+d_G(x,c_t)=j-l+1=d_C(c_t,c_l)$.

    If~$l-k\geq n/2$, then we have~$d_C(c_k,c_l)=d_C(c_k,c_t)+d_C(c_t,c_l)=d_G(c_k,c_t)+d_G(c_t,c_l)=d_G(x,c_k)+d_G(x,c_t)+d_G(x,c_l)+d_G(x,c_t)>d_G(c_k,c_l)$, violating the \ref{cond:banana} Condition, so~$l-k<n/2$ and hence~$d_C(c_k,c_l)=l-k$.
    We have~$l-k\leq d_C(c_k,c_l)\leq d_G(c_k,c_l)\leq d_G(c_k,x)+d_G(x,c_l)\leq o-k+l-o\leq l-k$, so~$d_C(c_k,c_l)=d_G(c_k,c_l)=d_G(c_k,x)+d_G(x,c_l)$.
    We get~$2(d_G(c_k,x)+d_G(c_l,x)+d_G(c_t,x))=d_G(c_k,c_l)+d_G(c_l,c_t)+d_G(c_t,c_k)\leq d_C(c_k,c_l)+d_C(c_l,c_t)+d_C(c_t,c_k)=n$ for distinct~$k$,~$l$, and~$t$, which violates the \ref{cond:ninja} Condition.  It follows that if a triangulation does not have a sketch, the \ref{cond:banana} or \ref{cond:ninja} Condition is violated.
  \end{proof}

\end{document}